\documentclass[10pt,journal,compsoc]{IEEEtran}

\ifCLASSOPTIONcompsoc
  \usepackage[nocompress]{cite}
\else
  \usepackage{cite}
\fi

\ifCLASSINFOpdf
  \usepackage[pdftex]{graphicx}
  \graphicspath{{./fig/}}
  \DeclareGraphicsExtensions{.pdf,.png,.eps}
\else
\fi

\usepackage{xcolor}

\usepackage{amsmath,amsthm}

\newtheorem{lemma}{Lemma}

\usepackage{algorithm}
\usepackage{algorithmicx}
\usepackage[noend]{algpseudocode}

\algrenewcommand\algorithmicrequire{\textbf{Input:}}
\algrenewcommand\algorithmicensure{\textbf{Output:}}

\usepackage{enumitem}

\usepackage{booktabs}
\usepackage{multirow}

\ifCLASSOPTIONcompsoc
 \usepackage[caption=false,font=footnotesize,labelfont=sf,textfont=sf]{subfig}
\else
 \usepackage[caption=false,font=footnotesize]{subfig}
\fi

\usepackage{url}

\usepackage[unicode]{hyperref}

\usepackage{xcolor}
\definecolor{paradisepink}{RGB}{233, 67, 94}

\begin{document}

\bstctlcite{bstctl:etal, bstctl:nodash, bstctl:simpurl}

\title{SEAnet: A Deep Learning Architecture for\\Data Series Similarity Search}

\author{Qitong~Wang,~Themis~Palpanas%
\IEEEcompsocitemizethanks{\IEEEcompsocthanksitem Qitong Wang and Themis Palpanas are affiliated with LIPADE, Université Paris Cité. E-mail: qitong.wang@etu.u-paris.fr, themis@mi.parisdescartes.fr.}
}

\markboth{Journal of \LaTeX\ Class Files,~Vol.~14, No.~8, August~2015}%
{Wang \MakeLowercase{\textit{and} Palpanas}: SEAnet: A Deep Learning Architecture for Data Series Similarity Search}

\IEEEtitleabstractindextext{%
\begin{abstract}
  A key operation for massive data series collection analysis is similarity search.
  According to recent studies, SAX-based indexes offer state-of-the-art performance for similarity search tasks.
  However, their performance lags under high-frequency, weakly correlated, excessively noisy, or other dataset-specific properties.
  In this work, we propose Deep Embedding Approximation (DEA), a novel family of data series summarization techniques based on deep neural networks.
  Moreover, we describe SEAnet, a novel architecture especially designed for learning DEA, that introduces the Sum of Squares preservation property into the deep network design.
  We further enhance SEAnet with SEAtrans encoder.
  Finally, we propose novel sampling strategies, SEAsam and SEAsamE, that allow SEAnet to effectively train on massive datasets.
  Comprehensive experiments on 7 diverse synthetic and real datasets verify the advantages of DEA learned using SEAnet %
  in providing high-quality data series summarizations and similarity search results.
  This paper was published in %KDD conference~\cite{c21-kdd-wang-seanet} and 
  the IEEE TKDE journal~\cite{DBLP:journals/tkde/WangP23}.
\end{abstract}

}

\maketitle

\IEEEdisplaynontitleabstractindextext

\IEEEpeerreviewmaketitle

\IEEEraisesectionheading{\section{Introduction}\label{sec:introduction}}

\IEEEPARstart{W}{ith} the rapid developments and deployments of modern sensors, massive data series\footnote{A data series, or data sequence, is an ordered sequence of points. 
The most common type of data series is time series, where the dimension that imposes the sequence ordering is time; though, this dimension could also be mass, angle, or position~\cite{j19-sigrec-Palpanas-report}.
} 
datasets are now being generated, collected and analyzed in almost every scientific domain~\cite{j19-sigrec-Palpanas-report}. 
Typical data series analysis techniques are querying~\cite{c19-vldb-Echihabi-return}, classification~\cite{DBLP:journals/datamine/FawazLFPSWWIMP20}, clustering~\cite{c15-sigmod-Paparrizos-kshape}, anomaly detection~\cite{sand}, and visualization~\cite{DBLP:journals/tvcg/GogolouTPB19}, for all of which similarity search plays a central role.
Data series similarity search aims to find the closest series in a dataset to a given query series according to a distance measure, 
such as Euclidean distance, which is one of the most widely used~\cite{DBLP:journals/datamine/WangMDTSK13}. 
Similarity search can be divided into exact search and approximate search~\cite{c18-vldb-Echihabi-lernaean}.
Approximate similarity search may not always produce the exact answers, but in most cases it produces answers that are very close to the exact ones~\cite{c19-vldb-Echihabi-return}.
Thus, it is %
widely used on massive series collections to enable interactive data exploration and other latency-bounded applications~\cite{pros}. 
In this work, we focus on approximate similarity search under Euclidean distance.

Indexes are widely employed to speed up data series similarity search~\cite{c18-vldb-Echihabi-lernaean,c19-vldb-Echihabi-return}.
Most indexes are based on summarized representations of the data series~\cite{DBLP:journals/datamine/WangMDTSK13} of lower dimensionality\footnote{In the data-series literature, \emph{dimensionality} is used interchangeably with \emph{length}, to refer to the number of values of a univariate series.}.
Symbolic Aggregate approXimation (SAX)~\cite{c08-kdd-Shieh-isax} %
is a popular and effective discretized summarization. %
SAX-based indexes~\cite{messijournal} are the state-of-the-art (SOTA) data series similarity search methods~\cite{c18-vldb-Echihabi-lernaean,c19-vldb-Echihabi-return}.
\begin{figure}[tb]
  \centering
  \subfloat[PAA works to approximate and reconstruct a RandWalk series]{
    \label{fig:motivation-rw}
    \includegraphics[width=.95\linewidth]{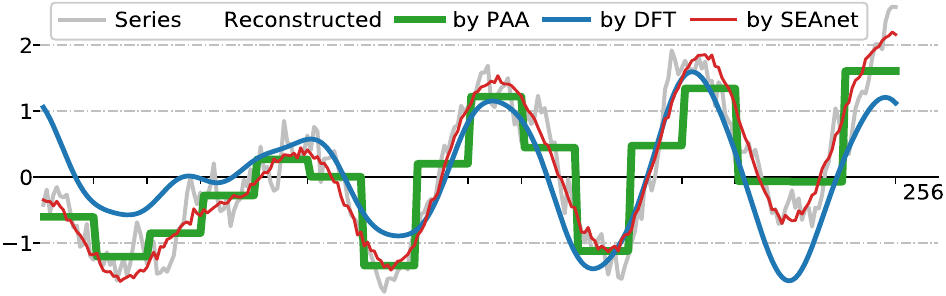}}
  \hfill
  \subfloat[PAA fails to approximate and reconstruct a Deep1B series]{
    \label{fig:motivation-deep}
    \includegraphics[width=.95\linewidth]{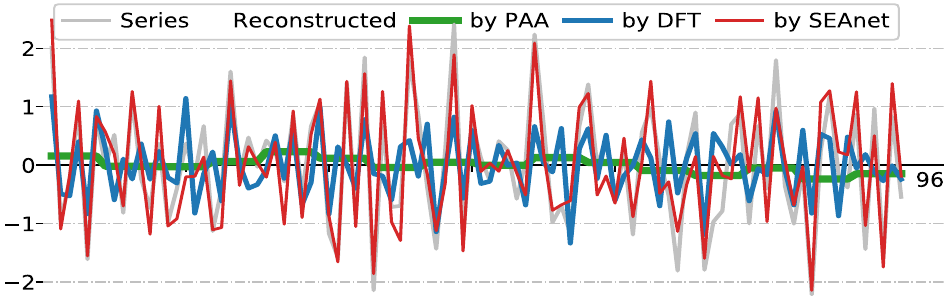}}
  \caption{Case studies where PAA and DFT work or fail to approximate and reconstruct series from RandWalk and Deep1B datasets. 
  In both cases, DEA works to approximate and reconstruct series.
  All summarizations use the same memory budget.}
  \label{fig:motivation}
\end{figure}

Nevertheless, SAX-based indexes
suffer from the problem that SAX fails in hard datasets with specific properties~\cite{j20-kais-Levchenko-bestneighbor}.
Since SAX is the symbolization of Piecewise Aggregate Approximation (PAA)~\cite{c08-kdd-Shieh-isax}, %
failure of PAA to correctly represent some data series directly translates to failure of the PAA-based SAX.
Figure~\ref{fig:motivation} illustrates a working and a failing case.
The high frequency of the Deep1B series (Figure~\ref{fig:motivation-deep}) implies more periodic intervals than the available SAX words: each PAA segment has to average values over $\ge$1 intervals, leading to similar PAA values across different segments, and to indistinguishable SAX words across different series.
Introducing more SAX words could alleviate the problem, but would lead to an undesirably long summarization that could not be effectively indexed. 

In this work, we propose to build a data series index based on \textbf{\textit{D}}eep \textbf{\textit{E}}mbedding \textbf{\textit{A}}pproximations (DEA), i.e., data series summarizations derived from embeddings learned using deep neural networks.
Embedding techniques, or representation learning, %
is to learn vectors possessing necessary latent information for %
downstream applications. 
Its success in data series, as well as high-dimensional vectors 
has been reported in many applications~\cite{DBLP:conf/wims/EchihabiZP20}.
Embedding techniques have been proven to be capable of capturing frequency~\cite{DBLP:conf/kdd/WangWLW18} and other latent properties.
However, to the best of our knowledge, data series embedding has not been adapted to and evaluated for similarity search.

\begin{figure}[tb]
  \centering
  \subfloat[PAA-based SAX symbolization]{
    \includegraphics[width=0.98\linewidth]{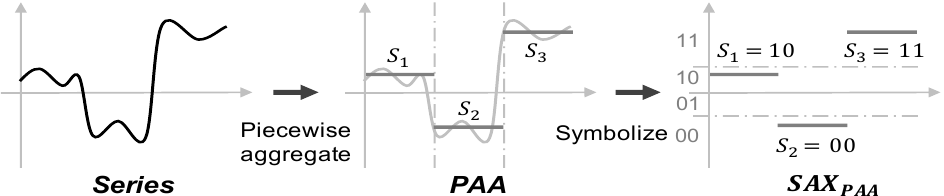}}
  \hfill
  \subfloat[DEA-based SAX symbolization]{
    \includegraphics[width=0.98\linewidth]{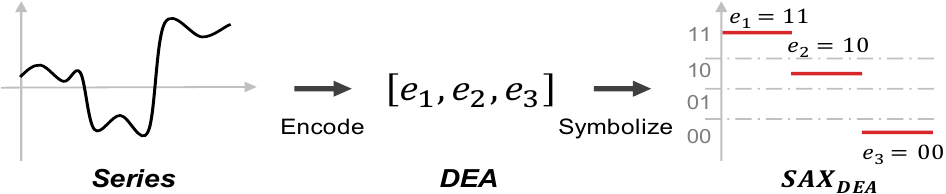}}
  \caption{Replace PAA by DEA for SAX symbolization.}
  \label{fig:intuition}
  \vspace*{-0.5\baselineskip}
\end{figure}

In the case of data series similarity search, DEA may replace PAA, and then be symbolized and indexed by an iSAX index as illustrated in Figure~\ref{fig:intuition}.
DEA targets to preserve original pairwise distances in the lower-dimensional DEA space.
Thus, it is naturally capable of being symbolized into SAX, on which an iSAX index can be built.
Our work shows that compared to PAA and (PAA-based) SAX, DEA better preserves pairwise distances, leading to a more effective index for data series similarity search.

To effectively learn DEA on massive series collections, we propose a novel autoencoder architecture SEAnet (\textbf{\textit{SE}}ries \textbf{\textit{A}}pproximation \textbf{\textit{net}}work).
SEAnet's basic structure follows a full-preactivation ResNet~\cite{DBLP:conf/eccv/HeZRS16}.
It adopts the idea of exponentially increasing dilations, which has been verified to be efficient for data series applications~\cite{DBLP:journals/corr/abs-1803-01271}.
  Moreover, we further design a new SEAtrans encoder, enhanced by Transformer blocks~\cite{c21-iclr-Dosovitskiy-vit} to improve the fixed global dependence framed by deeper dilated layers.
In contrast to existing convolutional autoencoders for series embedding~\cite{c19-neurips-Franceschi-embedding}, SEAnet comprises both an encoder and a decoder.
We argue (and experimentally verify) that a decoder is necessary to learn high-quality DEA for similarity search.
Moreover, SEAnet is the first architecture to formally introduce the principle of \textbf{\textit{S}}um \textbf{\textit{o}}f \textbf{\textit{S}}quares (SoS) preservation, and incorporate it into the network design. 
SoS preservation aims to keep the sum of squared values invariant throughout the transformations. 
We observe that defining new axes based on the largest SoS is equivalent to selecting the largest eigenvalues in eigenvalue-based linear dimensionality reductions on z-normalized datasets (i.e., mean=$0$, stddev=$1$)~\cite{j87-cils-Wold-pca}. 
They both aim at preserving the largest variances in the dataset through linear transformations.
In this sense, SoS could be regarded as an indicator of the quality of the transformation to a reduced dimensionality space performed by SEAnet (or other deep network architectures).
Hence, we introduce SoS as an invariant to regularize SEAnet and other networks, and demonstrate its benefits.

Finally, we observe that training a deep neural network on very large sequence collections is prohibitively expensive. 
Thus, for efficient training, we propose SEAsam (SEA-sampling), a novel sampling strategy based on a sortable data series summarization~\cite{DBLP:journals/pvldb/KondylakisDZP18}. 
SEAsam enables SEAnet (and other networks) to effectively fit %
a large dataset, leading to improved performance. 
Moreover, based on the observation that the raw data series are not the only sampling space to be preserved, we extend SEAsam by SEAsamE.
SEAsamE offers to SEAnet representative training samples covering all three major sampling spaces, i.e., the raw data series, the data series pairs, and the reconstruction errors.

Comprehensive experiments verify that, compared to PAA and DEA generated by other SOTA architectures (FDJNet~\cite{c19-neurips-Franceschi-embedding}, TimeNet~\cite{DBLP:conf/esann/MalhotraTVAS17} and InceptionTime~\cite{DBLP:journals/datamine/FawazLFPSWWIMP20}),
the DEA generated by SEAnet is more effective in preserving the original pairwise distances in the lower-dimensional summarized space.
This advantage also leads to more accurate approximate similarity search across several synthetic and real data series collections with diverse properties.

\noindent{{\bf [Contributions]}}
Our contributions\footnote{
  Compared to the previous conference version~\cite{c21-kdd-wang-seanet}, this paper describes an enhanced encoder architecture and a new sampling method.
} are as follows. 
\begin{enumerate}[noitemsep, topsep=0pt, wide=\parindent]
  \item We propose the use of deep learning embeddings for data series similarity search. 
  We introduce novel Deep Embedding Approximations, and show how these can be used to index the original data series and then support (approximate) similarity search queries.
  Our results can be used as a blueprint to facilitate further progress in this area.

  \item We propose SEAnet, a novel architecture that is specifically built to support high-quality DEA and similarity search. 
  SEAnet incorporates modern architectural elements designed for data series applications, including a full-preactivation ResNet and exponentially increasing dilations. 
    We extend SEAnet with a new SEAtrans encoder to provide learnable global dependence for deeper layers.

  \item We introduce and formalize the principle of Sum of Squares (SoS) preservation. 
  SoS preservation is a general principle for any architecture to learn high-quality DEA for dimensionality reduction.
  We explain how it can benefit the DEA architectures (including SEAnet), and how to incorporate it into the architecture designs.

  \item We propose SEAsam, a novel sampling strategy for massive data series collections, enabling effective training for the deep models.
  SEAsamE further extends SEAsam by exploiting three major sampling spaces in DEA learning
  and facilitates deep model training.
  
  \item We also describe alternative deep architectures for DEA, based on the SOTA designs of FDJNet, TimeNet, and InceptionTime. 
  We explain how our ideas can be applied on these architectures, and study in detail their performance.

  \item Comprehensive experiments on three synthetic datasets and four real-world datasets verified the effectiveness of DEA and SEAnet for data series summarization and approximate similarity search, outperforming traditional iSAX-based solutions, as well as three other SOTA RNN and CNN architectures for series embedding. 
  Datasets, codes and pre-trained models are available online\footnote{\url{https://helios.mi.parisdescartes.fr/~themisp/seanet/}}.
\end{enumerate}

\section{Related Work}\label{sec:literature}

\noindent{\bf [Data Series Indexes]}
The most prominent data series indexing techniques can be categorized into 
graph-based indexes~\cite{DBLP:journals/pami/MalkovY20},
inverted indexes~\cite{DBLP:journals/pami/BabenkoL15},
Locality Sensitive Hashing (LSH)~\cite{DBLP:journals/pvldb/HuangFZFN15},
optimized scans~\cite{DBLP:conf/cikm/FerhatosmanogluTAA00}, and tree-based indexes~\cite{c08-kdd-Shieh-isax,c13-vldb-Wang-dstree}. 
Recent studies~\cite{c18-vldb-Echihabi-lernaean, c19-vldb-Echihabi-return} have demonstrated that the SAX-based indexes~\cite{messijournal} achieve SOTA performance under several conditions.

SAX~\cite{c08-kdd-Shieh-isax}
is a discretized series summarization based on PAA. %
PAA first transforms the data series into $l$ real values, and then SAX quantizes each PAA value using discrete symbols. 
SAX quantization utilizes a scalar alphabet of size $a$.
The number of bits used to encode $a$ is called the cardinality of SAX.
iSAX~\cite{c08-kdd-Shieh-isax} enables the comparison of SAXs of different cardinalities, which makes SAX indexable. 
iSAX2.0 and iSAX2+~\cite{j14-kis-Camerra-isax2+} improve iSAX with a better node-splitting policy and a bulk-loading strategy, ADS+~\cite{j16-vldbj-Zoumpatianos-ads} makes iSAX adaptive, ULISSE\cite{j20-vldbj-Linardi-ulisse} supports to variable-length queries, DPiSAX~\cite{DBLP:journals/tkde/YagoubiAMP20} and Odyssey~\cite{odyssey} make iSAX distributed, ParIS~\cite{c18-bigdata-Peng-paris}, %
MESSI~\cite{messijournal} and SING ~\cite{sing} make iSAX highly concurrent, Dumpy~\cite{c23-sigmod-wang-dumpy} introduces an adaptive node splitting algorithm that leads to a multi-ary data structure, while Hercules~\cite{hercules} and Elpis~\cite{elpis} combine the iSAX and APACA~\cite{c13-vldb-Wang-dstree} summarizations.
Coconut-Trie and Coconut-Tree~\cite{j19-vldbj-Kondylakis-coconut} transform SAX into a sortable summarization to enable more bulk loading opportunities.

In this work, we employ MESSI~\cite{messijournal} as our indexing and query answering framework, because its memory-based concurrent design leads to SOTA performance. 
Moreover, since all SAX-based indexes share the same summarization techniques, improving the query answers quality of MESSI translates to the same improvements for all SAX-based indexes.

\noindent{\bf [Learned Data Series Embeddings]}
we observe that few recent works~\cite{DBLP:conf/esann/MalhotraTVAS17, c19-neurips-Franceschi-embedding} focus on data series representation learning, none of which targets similarity search.
Autoencoder is a category of deep neural networks to learn embeddings~\cite{c19-neurips-Franceschi-embedding}.
The encoder component of an autoencoder maps a dataset to lower dimensional vectors, i.e., embeddings; the decoder reverses this procedure.
It has been empirically verified in many domains that embedding learns useful latent information~\cite{DBLP:conf/kdd/WangWLW18}.

TimeNet~\cite{DBLP:conf/esann/MalhotraTVAS17} and FDJNet~\cite{c19-neurips-Franceschi-embedding} are two SOTA architectures for data series representation learning.
TimeNet deploys a multi-layer GRU %
to embed and reconstruct series.
FDJNet is based on Temporal Convolutional Network (TCN) %
to embed series.
Aside from representation learning, deep models are exploited in other data series applications.
The SOTA series classification method, InceptionTime~\cite{DBLP:journals/datamine/FawazLFPSWWIMP20} employs the Inception module for data series classification.
However, neither TimeNet, FDJNet nor InceptionTime has been adapted and evaluated for similarity search before.

In contrast to all the above methods, the proposed SEAnet not only adopts design choices suitable for distance preservation, but also introduces a novel and general principle of SoS preservation for dimensionality reduction.

Other data series embedding techniques that do not use neural networks include RWS~\cite{j18-arxiv-Wu-embedding}, SPIRAL~\cite{c19-ijcai-Lei-embedding}, and GRAIL~\cite{c19-Paparrizos-vldb-grail}; 
these are built upon various matrix factorization techniques,
e.g., Kernel Principal Component Analysis (KPCA).%
However, their high time and space complexities (mostly $\mathcal{O}(mn^2)$), prevents such techniques from being deploying on massive collections with hundreds of millions of data series, which is our goal.

\section{Background}\label{sec:background}

A \textbf{data series}, $S=\{p_1, ..., p_m\}$, is a sequence of points, where each point $p_i=(v_i,t_i)$, $1 \le i \le m$ is associated to a real value $v_i$ and a position $t_i$. 
The position corresponds to the order of this value in the sequence.  
We call $m$ the \emph{length} of the series. 
$\mathcal{S}$ denotes a collection of data series, i.e., $\mathcal{S}=\{S_1, ..., S_n\}$. 
We call $n$ the \emph{size} of the series collection.

A \textbf{summarization} $E=\{e_1, ..., e_l\}$ of a series $S$ is a lower, $l$-dimensional representation, which preserves some desired properties of $\mathcal{S}$.
For similarity search, the target property is pairwise distance space structure of $\mathcal{S}$, i.e., $\forall S_i, S_j \in \mathcal{S}, d'(E_i, E_j) \approx d(S_i, S_j)$, where $E_i$, $E_j$ are summarizations of $S_i$, $S_j$, $d(\cdot, \cdot)$, and $d'(\cdot, \cdot)$ are distance measures in series and summarization spaces, respectively.

The \textbf{distance measure} $d$ deployed in our work is Euclidean distance, which is a widely adopted and effective measure for data series similarity search~\cite{DBLP:journals/datamine/WangMDTSK13}.%
$d'$ in the summarization space needs not be the same as $d$, e.g., for PAA, $d'(\cdot, \cdot)=\frac{\sqrt{m}}{\sqrt{l}} \times d(\cdot, \cdot)$.
$d'$ for DEA is the same as PAA if it's scaled for SoS preservation.
Otherwise, $d'(\cdot, \cdot)=d(\cdot, \cdot)$.

Given a query series $S_q$ of length $m$, a series collection $\mathcal{S}$ of size $n$ and length $m$, a distance measure $d$, \textbf{similarity search} targets to identify the series $S_c \in \mathcal{S}$ whose distance to $S_q$ is the smallest, i.e., $\forall S_o \in \mathcal{S}, S_o \neq S_c, d(S_c, S_q) \leq d(S_o, S_q)$.
Instead of finding the exact answer $S_c$, \textbf{approximate similarity search} targets to find very fast a series $S_c' \in \mathcal{S}$, whose %
distance to $S_q$ is small, without guaranteeing it is the smallest.
$\frac{d(S_c, S_q)}{d(S_c', S_q)} \in (0,1]$ is called $S_c'$'s \textbf{tightness}.

\section{DEA-based Similarity Search}\label{sec:method}

In this section, we present the proposed DEA-based data series similarity search framework, including the SEAnet architecture. %
The complete workflow is illustrated in Figure~\ref{fig:method-workflow}.
Given a series collection, SEAsam first draws representative samples to train SEAnet.
After SEAnet converges, it embeds all series into DEAs, which are further discretized into SAXs.
Thus, DEA-based SAXs are used in an iSAX index, where approximate similarity search can be efficiently conducted.

\begin{figure}[tb]
  \centering
  \includegraphics[width=0.9\linewidth]{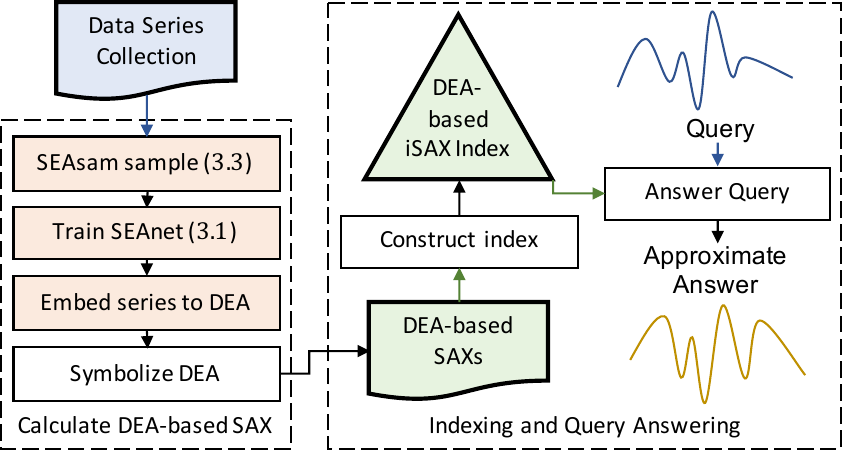}
  \caption{Workflow of DEA-based approximate similarity search.}
  \label{fig:method-workflow}
\end{figure}

\begin{figure*}[tb]
  \centering
  \subfloat[SEAnet Architecture\label{fig:method-arch}]{\includegraphics[width=.65\textwidth]{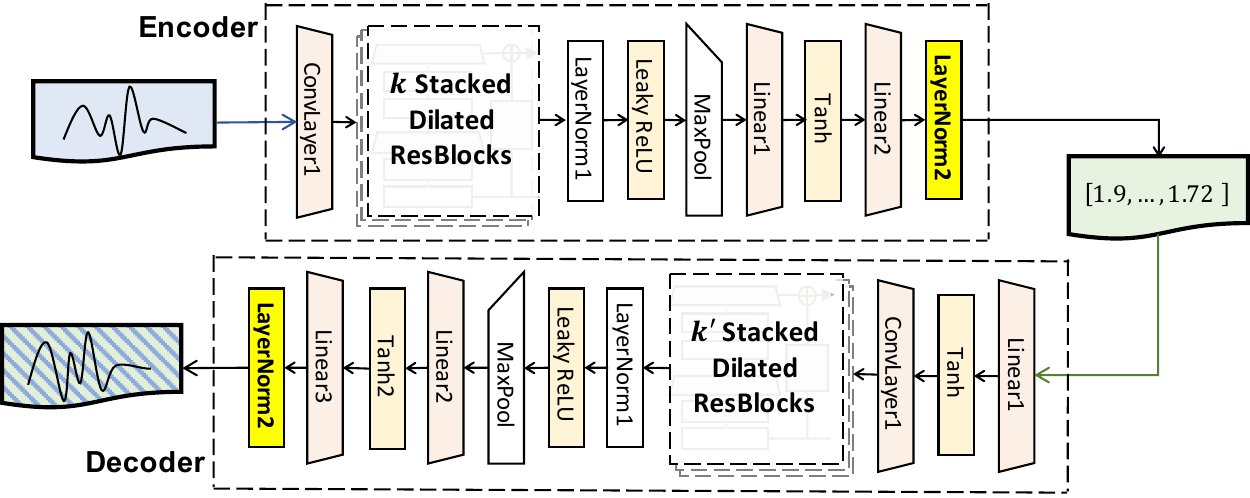}}
  \hspace{.05\textwidth}
  \subfloat[Dilated ResBlock in SEAnet\label{fig:method-block}]{\includegraphics[width=.25\textwidth]{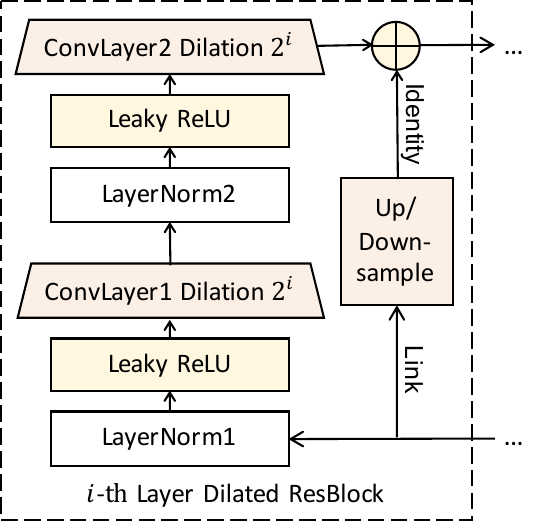}}
  \caption{The SEAnet architecture and the details of a dilated full-preactivation ResBlock.}
  \label{fig:arch-seanet}
\end{figure*}

SEAnet is a novel autoencoder proposed to learn high-quality DEA.
Unlike FDJNet~\cite{c19-neurips-Franceschi-embedding} and other SOTA convolutional architectures for series embedding,
SEAnet is composed of both an encoder and a decoder.
The inclusion of the decoder is beneficial, since it can act as a regularizer to prevent SEAnet from falling into bad local optima, where DEAs become very similar to each other (and hence, not suitable for similarity search).
SEAnet stacks dilated full-preactivation ResBlocks~\cite{DBLP:conf/eccv/HeZRS16}.
Its dilations increase exponentially with deeper layers.
Moreover, %
SEAnet introduces the principle of SoS preservation for lower dimensionality representation learning.
We present the details of SEAnet's design in Section~\ref{sec:architecture}, and further discuss the SoS preservation principle in Section~\ref{sec:ss}.

  To improve the fixed global dependence framed by the exponentially-increasing dilations, we propose a new SEAtrans encoder. %
  SEAtrans encoder replaces SEAnet's deeper ResBlocks by Transformer blocks (TransBlocks)~\cite{c21-iclr-Dosovitskiy-vit} to aggregate high-level information with learnable dependence.
  We present the SEAtrans encoder in Section~\ref{sec:arch-seatrans}.

Our SEAsam strategy makes use of the inverse iSAX sortable summarization~\cite{DBLP:journals/pvldb/KondylakisDZP18}.
In this scheme, SAX bits are interleaved, such that all significant bits across SAX words precede less significant bits, which renders the resulting representation, InvSax, sortable.
This order has been shown to imply the distribution information of the dataset~\cite{DBLP:journals/pvldb/KondylakisDZP18}.
Thus, sampling proceeds by drawing series of equal intervals from dataset sorted in InvSAX order (cf. Section~\ref{sec:sampling}).

SEAsamE further extends SEAsam by exploiting three major sampling spaces for DEA learning, i.e., the spaces of raw data series, data series pairs and reconstruction errors.
Hence, not only does it better represent the data series collection, but it also facilitates the convergence of SEAnet.
We present SEAsamE in Section~\ref{sec:seasam3}.

Since DEA acts as a replacement for PAA, %
indexing and query answering procedures remain similar %
to an ordinary iSAX.
Considering the low-latency requirement of approximate similarity search applications, our design follows MESSI~\cite{messijournal}, the SOTA concurrent in-memory iSAX index.

\subsection{\texorpdfstring{SEA\MakeLowercase{net}}{SEAnet} Architecture}\label{sec:architecture}

The SEAnet architecture is illustrated in Figure~\ref{fig:method-arch}.
It comprises a convolutional encoder and a homogeneous decoder.
We first overview the architecture and then present details.

The first part of the SEAnet encoder, from ConvLayer1 to MaxPool, comprises $k$ stacked dilated full-preactivation ResBlocks~\cite{DBLP:conf/eccv/HeZRS16} for nonlinear transformations.
The dilated ResBlock is illustrated separately in Figure~\ref{fig:method-block}.
Its dilation increases exponentially with the depth of the layer.
Compared to constant dilations, this has been verified to effectively broaden the receptive fields for data series applications~\cite{DBLP:journals/corr/abs-1803-01271}.
The dimension of latent vectors and number of channels are the same as the dimension of the input series.
Thus, after MaxPooling within channels and squeezing, the first part could be regarded as an equi-length nonlinear transformation.
The second part of the SEAnet encoder, from Linear1 to LayerNorm2, comprises two linear layers for dimensionality reduction.
Unlike most existing encoders with linear final layers~\cite{c19-neurips-Franceschi-embedding}, the SEAnet encoder is finalized by LayerNorm2, which is specifically designed using the SoS preservation principle. %
We elaborate on this in Section~\ref{sec:ss}.

The SEAnet decoder corresponds to the encoder, except for a preceding Tanh-activated linear layer, introduced to adjust dimensionality.
We claim that the encoder and decoder need not be homogeneous.
Although encoder-only architectures is the popular choice~\cite{c19-neurips-Franceschi-embedding}, we argue (and experimentally verify) that the decoder is necessary in similarity search applications in order to regularize the DEAs, so that they are distinguishable among each other. 
This results to a better indexing structure, and to a more effective and efficient similarity search.

\begin{figure}[tb]
  \centering
  \includegraphics[width=0.9\linewidth]{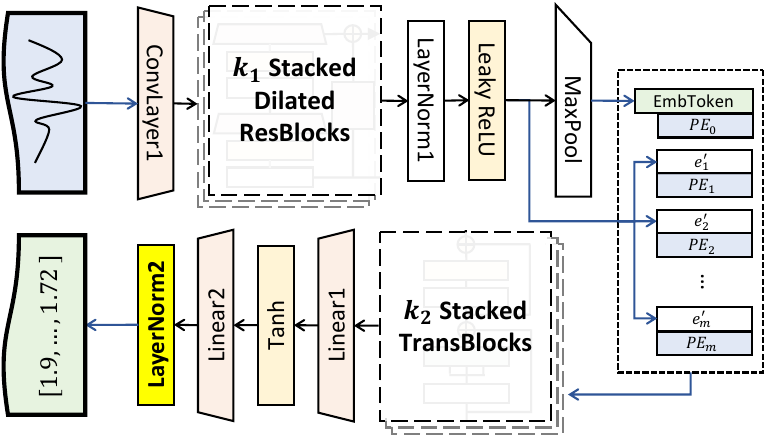}
  \caption{%
  The SEAtrans encoder architecture.}
  \label{fig:arch-seatrans}
\end{figure}

\noindent{\bf [Extending SEAnet with SEAtrans encoder]}\label{sec:arch-seatrans}
As shown in Figure~\ref{fig:arch-seatrans}, the SEAtrans encoder is composed by inserting $k_2$ stacked TransBlocks between the MaxPool and Linear1 layers.
A TransBlock consists of one multiheaded self-attention layer and one feedforward layer~\cite{c17-neurips-Vaswani-transformer}.
The input tokens to the stacked TransBlocks are the ResBlocks' output vectors split by dimensions, concatenated by learnable positional encodings~\cite{c21-kdd-Zerveas-mvts}.
Each token corresponds to a dimension.
Hence, each token represents the latent feature across all channels of each dimension.
In addition, we append a special embedding token (EmbToken) as the learning target, following the convention of classification token~\cite{c21-iclr-Dosovitskiy-vit}.

  The SEAtrans encoder complies to the SEAnet framework and is open to decoder architecture choices.
  SoS preservation (LayerNorm2 in Figure~\ref{fig:arch-seatrans}) and SEAsam sampling also work for SEAtrans.
  We use the SEAtrans encoder with the the regular SEAnet decoder. %

\noindent{\bf [Training Procedure]}\label{sec:training}
We first provide the intuitions behind SEAnet training, and then the mathematical formalization.

SEAnet is trained pairwisely by mini-batched Stochastic Gradient Descent (SGD).
Its loss function is a linear combination of two components: 
(1) The Compression Error $L_{C}$, i.e., the average differences between the original distance of data series pairs $(S_i, S_j)$ and their DEA distance.
$L_{C}$ evaluates whether original distances are well-preserved in the DEA space.
(2) The Reconstruction Error $L_{R}$, i.e., the average distance between the original series $S_i$ and the reconstructed series.
$L_{R}$ evaluates how well the original series can be reconstructed using SEAnet.
Moreover, we divide both the series and their DEAs by the square root of their lengths in $L_{C}$ and $L_{R}$.
Together with the LayerNorm2 in SEAnet, this scaling not only preserves SoS for better dimensionality reduction, but also stabilizes the gradient propagation.
The rationale behind these steps is further explained in Section~\ref{sec:ss}.

During each training epoch, for every series $S_i$ in the training set, a random but different series $S_j$ is drawn from the training set to form pairs $(S_i, S_j)$ for $L_{C}$.
Since both $S_i$ and $S_j$ are from the training set, $S_j$ is detached (%
treated as constants instead of %
variables) to prevent its gradients from being back-propagated twice within one epoch.

We now formalize $L_{C}$ and $L_{R}$. 
First, we introduce the formula of SEAnet encoder as $E_i = \phi(S_i \vert \Theta_\phi)$ and decoder as $\widetilde{S}_i = \psi(E_i \vert \Theta_\psi) = \psi\cdot\phi(S_i \vert \Theta_{\phi,\psi})$, where $\phi$ and $\psi$ are mappings with parameters $\Theta_\phi$ and $\Theta_\psi$, $S_i$ is a series, $E_i$ is $S_i$'s DEA and $\widetilde{S}_i$ is $S_i$'s reconstruction.
Without scaling, $L_{C} = \frac{1}{N_p} \sum_{(S_i, S_j) \in \mathcal{S} \times \mathcal{S}} \vert d(S_i, S_j) - d(\phi(S_i), \phi(S_j)) \vert$, where ${N_p}$ is the number of sampled series pairs $(S_i, S_j)$, and $L_{R} = \frac{1}{N_s} \sum_{S_i \in \mathcal{S}} d(S_i, \psi\cdot\phi(S_i))$, where ${N_s}$ is the number of sampled series $S_i$.
With scaling, $L_{C} = \frac{1}{N_p} \sum_{(S_i, S_j) \in \mathcal{S} \times \mathcal{S}} \vert \frac{1}{\sqrt{m}}d(S_i, S_j) - \frac{1}{\sqrt{l}} d(\phi(S_i), \phi(S_j)) \vert$, and $L_{R} = \frac{1}{N_s} \sum_{S_i \in \mathcal{S}} \frac{1}{\sqrt{m}} d(S_i, \psi\cdot\phi(S_i))$.
Thus, the loss function $L = L_{C} + \alpha L_{R}$, where $\alpha$ is a hyperparameter to balance between $L_{C}$ and $L_{R}$.

\subsection{Sum of Squares Preservation}\label{sec:ss}

We propose a SoS preservation framework for effective DEA learning.
SoS preservation has been observed before~\cite{j87-cils-Wold-pca}, but to the best of our knowledge, has never been formally introduced to representation learning.
Given an $n$$\times$$m$ matrix $M$, where %
row $M_{i,*}$ corresponds to a series and %
column $M_{*,j}$ corresponds to a position, SoS $=\sum_{i,j}M_{i,j}^2$.
Note that in linear dimensionality reductions on z-normalized datasets, defining new axes based on the largest SoS is equivalent to selecting the largest eigenvalues, with the purpose of preserving information (measuerd by total variance) about the dataset through linear transformations~\cite{j87-cils-Wold-pca}.
Thus, SoS could be regarded as a useful indicator of the transformation quality.
In SEAnet, by using SoS preservation as an invariant regularization, the quality of DEAs is upheld from this perspective, such that the neural network can focus on learning the nonlinear transformations. 

Given the (z-normalized) input dataset, the proposed SoS preservation requires two steps: (1) z-normalizing the output of encoder (DEAs) and decoder (the reconstructed series); and (2) dividing the series and their DEAs by the square of their lengths in loss function $L$.
Note that step (2) also takes the neural network convergence into consideration, as it benefits from the stabilization of the latent variables and variances~\cite{DBLP:journals/corr/BaKH16}.
We now elaborate on the design of SEAnet under this principle.

Considering the fact that z-normalizing data series is a very common operation~\cite{j19-sigrec-Palpanas-report},
we constrain SEAnet to keep SoS invariant by forcing each DEA to preserve the SoS of its corresponding series.
This is achieved by the following two steps: (1) z-normalizing the output of encoder, i.e., DEA; and (2) scaling DEA by $\sqrt{{m/l}}$.
We formalize this idea below, 
and prove it by Lemma~\ref{lem:ss}.

\begin{lemma}\label{lem:ss}
Given a z-normalized series dataset $\mathcal{S}$ of size $n$ and its DEAs $\mathcal{E}$, $\mathcal{E}'$ is derived by z-normalizing and then multiplying $\mathcal{E}$ by $\frac{\sqrt{m}}{\sqrt{l}}$.
$\mathcal{E}'$'s SoS is the same to $\mathcal{S}$, that is
{\footnotesize
\begin{equation}
\label{eq:scale}
 \sum_{S_i \in \mathcal{S}} \sum_{p^{i}_{j} \in S_i} {p^i_j}^2 = \sum_{E_i \in \mathcal{E}} \sum_{e^{i}_{j} \in E_i} (\frac{\sqrt{m}}{\sqrt{l}} \frac{e^i_j - \overline{e^i}}{\sigma_{e^i}})^2
\end{equation}
}
where $\overline{e^i}$ and $\sigma_{e^i}$ are the mean and standard deviation of DEA $e^i$ (without loss of generality, we assume $\sigma_{e^i}\neq 0.$).
\end{lemma}

\begin{proof}
\renewcommand{\qedsymbol}{}
First, as $\mathcal{S}$ is z-normalized with mean $=0$ and variance $=1$, the left side
{\footnotesize
\begin{equation*}
\begin{aligned}
\sum_{S_i \in \mathcal{S}} \sum_{p^{i}_{j} \in S_i} {p^i_j}^2 
&= \sum_{S_i \in \mathcal{S}} m \times \frac{\sum_{p^{i}_{j} \in S_i} (p^i_j - 0)^2}{m}\\
&= \sum_{S_i \in \mathcal{S}} m \times 1 = nm
\end{aligned}
\end{equation*}
}
For the right side of the equation,
{\footnotesize
\begin{equation*}
\begin{aligned}
\sum_{E_i \in \mathcal{E}} \sum_{e^{i}_{j} \in E_i} (\frac{\sqrt{m}}{\sqrt{l}} \frac{e^i_j - \overline{e^i}}{\sigma_{e^i}})^2
= &\sum_{E_i \in \mathcal{E}} \sum_{e^{i}_{j} \in E_i} \frac{m}{l} \frac{(e^i_j - \overline{e^i})^2}{\sigma_{e^i}^2}\\
= &\sum_{E_i \in \mathcal{E}} m \times \frac{1}{\sigma_{e^i}^2} \times \frac{\sum_{e^{i}_{j} \in E_i} (e^i_j - \overline{e^i})^2}{l}\\
= &\sum_{E_i \in \mathcal{E}} m \times 1 = nm
\end{aligned}
\end{equation*}
}
Thus, the conclusion holds that
{\footnotesize
\begin{equation*}
\sum_{S_i \in \mathcal{S}} \sum_{p^{i}_{j} \in S_i} {p^i_j}^2 = nm = \sum_{E_i \in \mathcal{E}} \sum_{e^{i}_{j} \in E_i} (\frac{\sqrt{m}}{\sqrt{l}} \frac{e^i_j - \overline{e^i}}{\sigma_{e^i}})^2
\end{equation*}
}
\end{proof}

\noindent{\bf [Scaling in Losses]}\label{sec:scaling}
Scaling the DEAs raises another problem: its values will have a much larger variance.
This hinders the convergence of the network, because of the internal covariate shift and other problems~\cite{DBLP:journals/corr/BaKH16}.
Latent variables with $\mu=0$ and $\sigma=1$ are widely considered among the best choices for the gradients' back-propagation~\cite{DBLP:journals/corr/BaKH16}.
Hence, we keep the DEAs z-normalized, and scale the series by $\sqrt{{1/m}}$
and DEA by $\sqrt{{1/l}}$
in $L_{C}$ and $L_{R}$.
This does not penalize SoS preservation, as shown in Equation~\ref{eq:scale-loss}, which is equivalent to Equation~\ref{eq:scale} by dividing $m$ on both sides.
{\footnotesize
\begin{equation}\label{eq:scale-loss}
 \sum_{S_i \in \mathcal{S}} \sum_{p^{i}_{j} \in S_i} (\frac{1}{\sqrt{m}}{p^i_j})^2 = n = \sum_{E_i \in \mathcal{E}} \sum_{e^{i}_{j} \in E_i} (\frac{1}{\sqrt{l}} \frac{e^i_j - \overline{e^i}}{\sigma_{e^i}})^2
\end{equation}
}

By further analyzing the distributions of distances under an ideal setting, we find that scaling series by $\sqrt{{1/m}}$
and DEAs by $\sqrt{{1/l}}$
in $L_{C}$ and $L_{R}$ would largely stabilize their pairwise distance distributions.
We formalize and prove this observation in Lemma~\ref{lem:loss}.

\begin{lemma}\label{lem:loss}
 Given two series $S_1=\{p^1_1, ..., p^1_m\}$ and $S_2=\{p^2_1, ..., p^2_m\}$.
 Ideally, assume $p^1_1, ..., p^1_m, ..., p^2_1, ..., p^2_m$ are i.i.d. $\sim N(0,1)$.
 Scaling $S_1$ and $S_2$ by $\frac{1}{\sqrt{m}}$ will reduce the mean and variance of $d(S_1, S_2)$ by $\sqrt{m}$ and $m$, respectively.
\end{lemma}
\begin{proof}
\renewcommand{\qedsymbol}{}
As $\forall i \in [1, ..., m]$, $p^1_i$ and $p^2_i$ are i.i.d. $\sim N(0,1)$, we have
{\footnotesize
\begin{equation*}
  \begin{aligned}
  p^1_i - p^2_i \sim N(\mu_1 - \mu_2, \sigma_1^2 + \sigma_2^2) &\sim N(0, 2)\text{,}\\
  d(S_1, S_2) = \sqrt{\sum_{i=1}^m (p^1_i - p^2_i)^2} &\sim \chi_m(\sqrt{2})
\end{aligned}
\end{equation*}
}
where $\chi_m(\sqrt{2})$ is a centered $\chi$ distribution with $m$ degree of freedom, derived from i.i.d. variables $\sim N(0, 2)$.

For $x \sim \chi_m(\sqrt{2})$, we have
{\footnotesize
\begin{equation*}
  \begin{aligned}
  \mu_{\chi_m(\sqrt{2})} &= \sqrt{2} \frac{\Gamma(\frac{m + 1}{2})}{\Gamma(\frac{m}{2})} \times \sqrt{2}\text{,} \\
  \sigma^2_{\chi_m(\sqrt{2})} &= 2 (\frac{m}{2} - (\frac{\Gamma(\frac{m + 1}{2})}{\Gamma(\frac{m}{2})})^2) \times 2
  \end{aligned}
\end{equation*}
}

Similarly, dividing $p^1_i$ and $p^2_i$ by $\sqrt{m}$, we have
{\footnotesize
\begin{equation*}
  \begin{aligned}
    \frac{p^1_i}{\sqrt{m}} \text{ and } \frac{p^2_i}{\sqrt{m}} \sim N(0, \frac{1}{m}) &\text{, } \frac{p^1_i}{\sqrt{m}} - \frac{p^2_i}{\sqrt{m}} \sim N(0, \frac{2}{m})\text{,}\\
    d(\frac{S_1}{\sqrt{m}}, \frac{S_2}{\sqrt{m}}) &\sim \chi_m(\sqrt{\frac{2}{m}})
\end{aligned}
\end{equation*}
}

And for $x \sim \chi_m(\sqrt{\frac{2}{m}})$,
{\footnotesize
\begin{equation*}
  \begin{aligned}
  \mu_{\chi_m(\sqrt{\frac{2}{m}})} &= \sqrt{2} \frac{\Gamma(\frac{m + 1}{2})}{\Gamma(\frac{m}{2})} \times \sqrt{\frac{2}{m}}\text{,} \\
  \sigma^2_{\chi_m(\sqrt{\frac{2}{m}})} &= 2 (\frac{m}{2} - (\frac{\Gamma(\frac{m + 1}{2})}{\Gamma(\frac{m}{2})})^2) \times \frac{2}{m}
  \end{aligned}
\end{equation*}
}

Thus, the conclusion holds that 
{\footnotesize
\begin{equation*}
  \begin{aligned}
  \frac{\mu_{\chi_m(\sqrt{\frac{2}{m}})}}{\mu_{\chi_m(\sqrt{2})}} = \frac{1}{\sqrt{m}} \text{, }
  \frac{\sigma^2_{\chi_m(\sqrt{\frac{2}{m}})}}{\sigma^2_{\chi_m(\sqrt{2})}} = \frac{1}{m}
  \end{aligned}
\end{equation*}
}
\end{proof}
  
Typical examples of scaling %
are presented in Table~\ref{tab:loss-scaling}. 
We make three observations:
(1) After scaling short series by $\sqrt{{256/m}}$,
the means of distance distributions are comparable to series of length 256.
This confirms that our design of SoS preservation is indeed helpful to preserve pairwise distances.
(2) However, the variance of distance distributions increases dramatically after scaling, e.g., 16$\times$ from 0.984 to 15.743 for length 16.
This introduces extra noise that hinders convergence.
(3) By scaling both series and DEAs in the loss functions, not only are the means of distance distributions kept roughly the same ($\approx$1.4), but also their variances are suppressed to a small value.
This helps SEAnet focus on learning from the differences between series distance and DEA distance, without being interfered by endogenous noises of the distance distributions.

\begin{table}[tb]
  \centering
  \setlength\tabcolsep{3.85pt}
  \caption{Mean and Variance of the distribution of pairwise distances between two ideal series.
  Scaling by $\sqrt{{256}/{m}}$ is to preserve SoS by scaling DEA itself. 
  Scaling by $\sqrt{{1}/{m}}$ is the case where we scale both series and DEA in loss functions.
  (Note that for length $256$, scaling by $\sqrt{{256}/{m}}$ does not change the original behavior.)}
  \label{tab:loss-scaling}
  \begin{tabular}{@{}c|cc|cc|cc@{}} 
    \toprule
    \multicolumn{1}{c}{Length} & \multicolumn{2}{c}{Before Scaling} & \multicolumn{2}{c}{$\times \sqrt{{256}/{m}}$} & \multicolumn{2}{c}{$\times \sqrt{{1}/{m}}$} \\ 
    \cmidrule(rl){2-3} \cmidrule(rl){4-5} \cmidrule(rl){6-7}
    \multicolumn{1}{c}{$m$} & Mean & \multicolumn{1}{c}{Var} & Mean & \multicolumn{1}{c}{Var} & Mean & Var \\
    \midrule
    256 & 22.605 & 0.999 & 22.605 & 0.999 & 1.4128 & 0.0039\\
    128 & 15.969 & 0.998 & 22.583 & 1.9961 & 1.4115 & 0.0078 \\
    96 & 13.820 & 0.997 & 22.569 & 2.6597 & 1.4105 & 0.0104 \\
    16 & 5.5692 & 0.984 & 22.277 & 15.743 & 1.3923 & 0.0615 \\
    8 & 3.8772 & 0.967 & 21.933 & 30.944 & 1.3708 & 0.1209 \\
    \bottomrule
  \end{tabular}
\end{table}

Finally, we observe that scaling series and DEA will not only keep the two distances to the same level, but will also largely stabilize the distance distributions.
Both effects are beneficial to SEAnet's learning and convergence.
Thus, by z-normalizing DEA, and scaling series and DEA in $L_{C}$ and $L_{R}$, SEAnet succeeds in providing high-quality DEAs by preserving SoS, and in converging fast to good optima (thanks to the stable latent variables and gradients).

\subsection{Sampling with SEAsam and SEAsamE}\label{sec:sampling}

The representativeness of the training set upper bounds the quality of the deep models.
To effectively train SEAnet on very large ($\geq$$1e8$) series collections, a good sampling strategy is essential for providing representative subsets.
This means that the sample should effectively cover the entire space of a given dataset, and we need to efficiently select this sample without having to perform expensive computations on the full dataset. 
For example, an effective but prohibitively expensive sampling strategy would be to sample from all leaves of an index built on the entire dataset, since such an index would cover the entire space, and each leaf would gather similar series.

To this end, we propose SEAsam (SEA-sampling), a novel data series sampling strategy based on the sortable data series representation, InvSAX~\cite{DBLP:journals/pvldb/KondylakisDZP18}. 
Recall that SAX first transforms the data series into $l$ real values (i.e., the mean values of $l$ segments of consecutive points of the series), and then quantizes these real values, representing them using discrete symbols (usually of cardinality 256)~\cite{c08-kdd-Shieh-isax}. 
The core observation is that every subsequent bit in a SAX word contains a decreasing amount of information about the location of its corresponding data point, and simply increases the degree of precision. %
Interleaving SAX's bits such that all significant bits across each SAX word precede all less significant bits presents a value array with descending significance, i.e., InvSAX.

The procedure to generate InvSAX is shown in Figure~\ref{fig:coconut}.
The most significant bits \{$1,1,0,0$\} across the original SAX words are moved to the first 4 bits of InvSAX, making its first and most significant value $6$ (shown in red/bold). 
The second most significant bits \{$1,0,1,1$\} are moved to InvSAX's 5-8 bits, making the second value $2$.
The last bits \{$0,1,1,0$\} are moved to InvSAX's 9-12 bits, making the last two values $2$ and $6$.
Thus, this series will be order by its InvSAX [$6,2,2,6$].

\begin{figure}[tb]
  \centering
  \advance\leftskip-0.05\linewidth
  \includegraphics[width=0.7\linewidth]{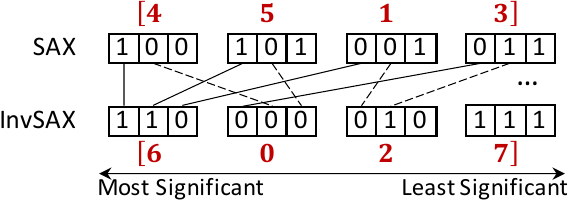}
  \caption{SEAsam transformation and InvSAX~\cite{DBLP:journals/pvldb/KondylakisDZP18}.}
  \label{fig:coconut}
\end{figure}

As illustrated in Figure~\ref{fig:seasam}, SEAsam orders the series collection by their InvSax representations, and draws samples at equal-intervals (e.g., every 1,000 series) from this sorted order.
Thus, SEAsam samples are expected to preserve the distribution of the series collection by evenly covering its InvSAX space.
Moreover, the time complexity of SEAsam is $\mathcal{O}(nm)$, and the space complexity of SEAsam is $\mathcal{O}(nl)$, rendering SEAsam an efficient strategy.

\begin{figure}[tb]
  \centering
  \advance\leftskip-0.05\linewidth
  \includegraphics[width=0.8\linewidth]{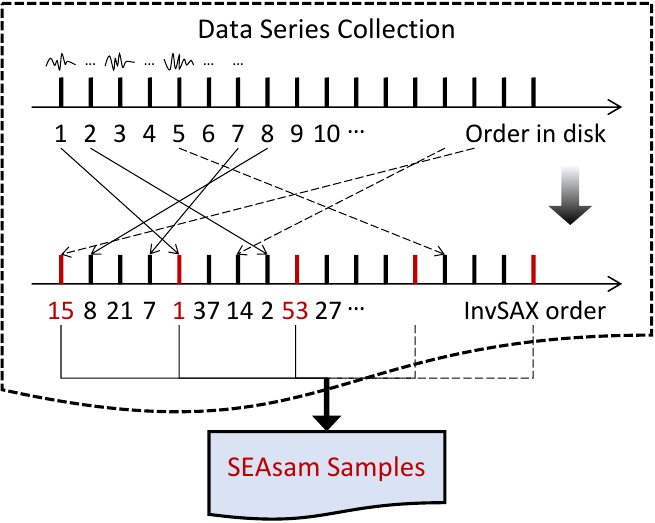}
  \caption{SEAsam sampling workflow.}
  \label{fig:seasam}
\end{figure}

Algorithm~\ref{algo:sampling} provides the pseudo-code of SEAsam sampling.
Line 1 generates SAXs for series collection $\mathcal{S}$.
Lines 3-8 interleave SAX bits to generate InvSAX.
This transformation could be implemented in place if under memory pressure, i.e., replacing SAX by its corresponding InvSAX.
Lines 9-10 sort InvSAX, select $n'$ equal-interval indices and return the corresponding sampled series.
After sampling from the data series collection, we randomly permute the samples to make training and validation sets during each epoch of autoencoder training.

\begin{algorithm}[tb]
{\small
  \caption{SEAsam Sampling}
  \label{algo:sampling}
  \begin{algorithmic}[1]
  \Require data series collection $\mathcal{S}$ of size $n$, sample size $n'$
  \Ensure sample set $\mathcal{S}'$
  \State $\hat{\mathcal{S}}$ = \Call{SAXTransform}{$\mathcal{S}$}
  \State Inv$\hat{\mathcal{S}} = \emptyset$
  \ForAll{$\hat{S} \in \hat{\mathcal{S}}$}
    \State Inv$\hat{S}$ = $0_{SAX}$
    \ForAll{$i \in$ \Call{range}{SAX cardinality}}
      \ForAll{SAX symbol $\hat{s} \in \hat{S}$}
        \State Inv$\hat{S}$.append(the $i$ bit of $\hat{s}$)
      \EndFor
    \EndFor
    \State Inv$\hat{\mathcal{S}}$.add(Inv$\hat{S}$)
  \EndFor
  \State indices = \Call{Sort}{Inv$\hat{\mathcal{S}}$}.indices[: : $n/n'$]
  \State $\mathcal{S}'=\mathcal{S}$[indices]
  \end{algorithmic}
  } %
\end{algorithm}

\noindent{\bf [Extending SEAsam with SEAsamE]}\label{sec:seasam3}
The empirical effectiveness of SEAsam sampling (cf. Section~\ref{fig:exp-sampling}) and the importance of the training set quality for our problem motivate the need to develop better and tailored sampling strategies. 
Designing sampling strategies demands a close investigation into the sampling space.
Although SEAsam sampling is designed to represent well the raw data series space, a key observation is that there are more spaces to be explored for DEA learning.
We first provide a discussion of the different sampling spaces in DEA learning.
On this basis, a novel combined sampling strategy, SEAsamE, is proposed in order to provide a combined use for all these spaces.

Our main observation is that there are three major sampling spaces in DEA learning, which could be organized into two categories.

The first category regards the dataset-specific sampling spaces, including the raw data series space targeted by SEAsam.
Another important sampling space falling into this category is the distance distribution of all possible series pairs in the collection, which is exactly what we would like SEAnet to preserve in the lower-dimensional DEA space.
From the perspective of deep model training, these two spaces correspond to the distributions of raw features and targeted values.
Preserving both distributions has been deemed beneficial in previous studies for capturing the ground truth and avoid bias in deep learning~\cite{c20-kdd-Wilson-domainadaptation}.

The second category regards the model-specific sampling spaces.
Intuitively speaking, these spaces represent model's response to the input features, i.e., the difficulty to embed and reconstruct a series.
In DEA learning, the reconstruction errors serve this purpose.
Unlike preserving the same distributions for the dataset-specific sampling spaces, different strategies facilitate model training following different heuristics~\cite{c17-iccv-wu-weighted}, 
e.g., sampling harder series is expected to offer faster convergence. %
In this work, we propose to draw samples evenly in terms of their reconstruction error values (instead of reconstruction error frequencies). 
That is, the sample set has a uniform distribution of reconstruction errors.
This sampling strategy %
is capable to correct the bias while controlling the gradient variance~\cite{c17-iccv-wu-weighted}.

In order to provide a training set that represents all three major sampling spaces mentioned above, we propose a novel composite sampling strategy SEAsamE. %
SEAsamE incorporates SEAsam to preserve the raw series space, with extra steps for the other two spaces.
An extra preprocessing step is introduced to collect an empirical distribution of the pairwise distances $\bar{P}(d(S_i,S_j))$.
To draw a set of training series, SEAsamE first draws a candidate set of SEAsam samples.
Second, reconstruction errors $L_{R}(\mathcal{S})$ of the candidate series are calculated.
The training set is drawn based on the inverse of the frequencies of reconstruction errors $1 / \bar{P}(L_{R}(S_i))$.
To draw series pairs for the calculation and backpropagation of $L_{C}(\mathcal{S}, \mathcal{S})$, distances of all series pairs within each mini-batch are calculated.
Then, SEAsamE selects series pairs according to their estimated distance probabilities $\bar{P}(d(S_i,S_j))$.

\begin{algorithm}[tb]
{\small
  \caption{Training with SEAsamE}
  \label{algo:seasam3}
  \begin{algorithmic}[1]
  \Require deep model with encoder $\phi$ and decoder $\psi$ parameterized by $\Theta$, data series collection $\mathcal{S}$, training size $n'$ and candidate size $n'_c$, batch size $b_s$ for series and $b_p$ for pairs, max epochs max$e$
  \Ensure learned parameters $\Theta$
  \State estimate $\bar{P}(d)$
  \ForAll{$i \in$ \Call{range}{max$e$}}
    \If{\Call{SampleCheck}{$i$,$\Theta$}}
      \State $\mathcal{S}'=$\Call{SampleSeries}{$\mathcal{S}$, $\Theta$, $n'$, $n'_c$}
    \EndIf
    \ForAll{mini-bach $\mathcal{S}'_b \in \mathcal{S}'$}
      \State $\mathit{E}'_b=\phi(\mathcal{S}'_b)$, $\mathcal{R}'_b=\psi(\mathit{E}'_b)$
      \State $\mathcal{S}'_p=$\Call{SamplePairs}{$\mathcal{S}'_b$, $\bar{P}(d)$, $b_p$}
      \State $L=L_C(\mathcal{S}'_p , \mathit{E}'_p)+\alpha L_R(\mathcal{S}'_b , \mathcal{R}'_b)$
      \State backpropagate $\nabla L$
    \EndFor
  \EndFor

  \vspace{-.15cm}
  \Statex
  \hrulefill
  \renewcommand{\algorithmicrequire}{\textbf{Function}}
  \Require \Call{SampleSeries}{$\mathcal{S}$, $\Theta$, $n'$, $n'_c$}
  \State $\mathcal{S}'_c=$\Call{SEAsam}{$\mathcal{S}$, $n'_c$}
  \State $\mathcal{R}'_c=\psi\cdot \phi(\mathcal{S}'_c)$
  \State estimate $\bar{P}(L_R)$ according to $L_R(\mathcal{S}'_c , \mathcal{R}'_c)$
  \State draw $\mathcal{S}'$ of $n'$ series $\sim \bar{P}(L_R)$ \\
  \Return $\mathcal{S}'$
  
  \vspace{-.1cm}
  \Statex
  \hrulefill
  \Require \Call{SamplePairs}{$\mathcal{S}'_b$, $\bar{P}(d)$, $b_p$}
  \State calculate $d(\mathcal{S}'_b, \mathcal{S}'_b)$
  \State draw $\mathcal{S}'_p=\{(\mathcal{S}'_i, \mathcal{S}'_j)\}$ of $b_p$ pairs $\sim \bar{P}(d)$ \\
  \Return $\mathcal{S}'_p$
  \end{algorithmic}
} %
\end{algorithm}

Algorithm~\ref{algo:seasam3} shows how to train SEAnet with SEAsamE.
Line 1 is a preprocessing step to estimate $\bar{P}(d)$.
(This step can be implemented by calculating histograms based on a sufficient number of random samples.)
At the beginning of each epoch, line 3 checks whether to invoke function \textsc{SampleSeries}, and then line 4 draws a new series sample set.
\textsc{SampleSeries} draws $n'_c$ SEAsam samples $\mathcal{S}'_c$ in line 10, calculates their reconstruction errors in line 11-12, and returns $n'$ samples from $\mathcal{S}'_c$ according to $\bar{P}(L_R)$.
Lines 5-6 consider each mini-batch and calculate their DEAs and reconstructions.
To draw $b_p$ series pairs, line 7 invokes function \textsc{SamplePairs}, which calculates distances for all series pairs in line 15, and returns $b_p$ sampled pairs according to $\bar{P}(d)$.
Last, lines 8-9 calculate and backpropagate the loss $L$.

\section{Experimental Evaluation}\label{sec:experiments}

We present our experimental evaluation of SEAnet, DEA-based data series similarity search, SEAsam and SEAsamE using 7 diverse synthetic and real datasets. 
In summary, the results demonstrate that the SEAnet DEA is robust across various dataset properties and outperforms its competitors by better preserving original pairwise distances and nearest neighborhood structure, leading to better approximate similarity search results than traditional (PAA-based) and alternative deep learning (DEA-based using FDJNet, TimeNet, and InceptionTime) approaches.

\noindent{\bf [Setup]}\label{sec:exp-setup}
All deep models were trained using Nvidia Tesla V100 SXM2 (16G memory).
Sampling, indexing and query answering were conducted in a server with 2x Intel(R) Xeon(R) Gold 6134 CPU @ 3.20GHz and 320GB RAM. 
Software environments were python/3.6.10, pytorch-gpu/py3/1.5.1 and cuda/10.2.

\noindent{\bf [Datasets]}\label{sec:exp-dataset}
Experiments were conducted on 3 synthetic datasets of different characteristics and 4 real datasets from diverse domains.
For synthetic datasets, we used RandWalk, F5 and F10. 
RandWalk~\cite{c18-vldb-Echihabi-lernaean,c19-vldb-Echihabi-return} was generated as cumulative sums of steps following a standard Gaussian distribution $N(0,1)$. 
F5 and F10 were recently introduced to evaluate iSAX on datasets of different frequencies~\cite{j20-kais-Levchenko-bestneighbor}.
They were generated through Inverse Discrete Fourier Transform (IDFT) from a random spectrum with its first 5 or 10 components being amplified.
The $4$ real datasets are Seismic from seismology, Astro from astronomy, SALD from neuroscience and Deep1B from image processing~\cite{c19-vldb-Echihabi-return}.
Length of each series is 128 for SALD, 96 for Deep1B, and 256 for the rest.
Note that these datasets are considered hard for similarity search~\cite{c18-vldb-Echihabi-lernaean,c19-vldb-Echihabi-return}.
We experimented with dataset sizes between 1M to 100M series
(100M series of length 256 $\approx$100GB).

\noindent{\bf [Methods]}\label{sec:exp-methods}
We evaluated the SEAnet-generated DEA and its applications in data series similarity search against PAA and DEA generated by SEAnet-nD (a simplified version of SEAnet), and our adaptations of FDJNet~\cite{c19-neurips-Franceschi-embedding}, TimeNet~\cite{DBLP:conf/esann/MalhotraTVAS17}, and InceptionTime~\cite{DBLP:journals/datamine/FawazLFPSWWIMP20}.
We describe these methods below.
SEAsam was compared to uniformly random sampling. 

PAA %
is the baseline method to evaluate DEA's summarization quality.
PAA-based MESSI iSAX index~\cite{messijournal} is the SOTA baseline method to evaluate DEA-based iSAX on approximate similarity search.
When it is clear from the context, we used the term PAA to denote both the PAA summarization, and the PAA-based iSAX index in the rest of Section~\ref{sec:experiments}.
We used the same convention for other methods, as well, e.g., SEAnet denotes the DEA generated by SEAnet and the index built on the DEA generated by SEAnet.

SEAnet-nD is an encoder-only version of SEAnet, introduced to evaluate the contribution of the decoder to the final performance.
The convolution kernel size for SEAnet and SEAnet-nD is 3.
For FDJNet, we adopted the same network setting as SEAnet-nD.
For TimeNet, we used the output of the last position as DEA, instead of the original concatenation of latent vectors.
This enables TimeNet using longer latent vectors to generate lower-dimensional DEA.
The dropout probability of TimeNet is 0.4, following the original setting~\cite{DBLP:conf/esann/MalhotraTVAS17}.
For InceptionTime, we used the same structure as SEAnet, but replaced ResBlock with InceptionBlock~\cite{DBLP:journals/datamine/FawazLFPSWWIMP20}.
The convolution kernel sizes for InceptionTime are \{3, 5, 9, 17\}.
Batch size was set to 128 for TimeNet (due to our memory limit), and to 256 for the other architectures.
For all architectures, we stacked 7, 6, and 5 building blocks for series of length 256, 128, and 96, respectively.
Latent dimensions and channels were the same to series length. %

All models were trained using SGD and the same loss function (cf. Section~\ref{sec:training}).
Training size was 200,000 series, and validation size was 20,000.
TimeNet was trained for 125 epochs, while the others for 100 epochs.
Hyperparameters were tuned for each model (of specific DEA length) on 100M datasets.
The best hyperparameters for similarity search were adopted for all other dataset sizes.
$\alpha$ was searched from \{$0.1$, $0.25$, $0.5$, $1$, $1.25$\}.
Learning rate was cross searched from \{$1e\text{-}3$, $5e\text{-}3$, $1e\text{-}2$, $2.5e\text{-}2$, $5e\text{-}2$\}, and was either linearly decayed (every epoch), or exponentially decayed (by $0.9$ every $2$ epochs) until $1e\text{-}5$.
Totally, 6,090 deep models were trained to provide a thorough profile of DEA architectures.
Other hyperparameters were set to their default values.
For indexing, leaf size $h$ was 8,000 by default.

SEAnet encoder architecture details are reported in Table~\ref{tab:model-specifics}.
We use the respective symmetrical architectures as their decoders.
The complexities were measured in number of (millions of) tunable parameters, and (billions of) Multiply-and-Accumulate (MAC) operations.
All models converged with the training speed reported in Figure~\ref{fig:exp-convergence}.

\begin{table}[tb]
  \centering
  {\small
  \caption{%
    SEAnet encoder architecture specifics for different series lengths. Building blocks ([kernel size, channel number] for the convolutional layer/block, [dimension] for the linear layer) are shown in brackets, multiplied by the numbers of blocks stacked.
  }
  \label{tab:model-specifics}
  \begin{tabular}{@{}c|lll@{}} 
      \toprule
    layer name & 13-layer (\emph{96}) & 15-layer (\emph{128}) & 17-layer (\emph{256}) \\
    \midrule
    ConvLayer & $\begin{bmatrix} 3, 96 \end{bmatrix}$ & $\begin{bmatrix} 3, 128 \end{bmatrix}$ & $\begin{bmatrix} 3, 256 \end{bmatrix}$ \\ [3pt]
    ConvBlock & $\begin{bmatrix} 3, 96\\ 3, 96 \end{bmatrix}\times$5 & $\begin{bmatrix} 3, 128\\ 3, 128 \end{bmatrix}\times$6 & $\begin{bmatrix} 3, 256\\ 3, 256 \end{bmatrix}\times$7 \\ [7.5pt]
    Linear & $\begin{bmatrix} 96\\ 16 \end{bmatrix}$ & $\begin{bmatrix} 128\\ 16 \end{bmatrix}$ & $\begin{bmatrix} 256\\ 16 \end{bmatrix}$ \\
    \midrule
    parameters/M & 0.585 & 1.234 & 5.712 \\
    MACs/B & 0.053 & 0.152 & 1.412 \\
    \bottomrule
  \end{tabular}
  } %
\end{table}

\noindent{\bf [Measures]}\label{sec:exp-measures}
To evaluate summarization quality, we used three measures: average distance differences (unscaled $L_{C}$), reconstruction RMS, and NN coverage.
Series subsets, or series pairs are SEAsam samples from 100M datasets.

\begin{enumerate}[noitemsep,topsep=0pt,wide, labelwidth=!, labelindent=0pt]
  \item \textbf{Average Distance Differences} (\textbf{$L_{C}$}). 
  Differences between original distances and DEA distances of series pairs, i.e., $\vert d'(\phi(S_i), \phi(S_j)) - d(S_i, S_j) \vert$.
  Reported values were averaged from 20,000 pairs.
  (Differences of 1,000 pairs were illustrated as scatters in Figure~\ref{fig:exp-scatter}.)

  \item \textbf{Reconstruction RMS}. 
  Root-Mean-Square errors between original series and their reconstructions, i.e., $\sqrt{\frac{1}{m}\sum_{i}(p_i-p'_i)^2}$, where $S'=[...,p'_i,...]$ is the reconstruction of series $S=[...,p_i,...]$.
  Reported values were averaged from 20,000 series.

  \item \textbf{NN Coverage}.
  The coverage of series $S$'s nearest neighbors in DEA space, i.e., $\frac{\vert k\mathrm{NN}_{d}(S) \cap k\mathrm{NN}_{d'}(E)\vert}{\vert k\mathrm{NN}_{d}(S)\vert}$, where $k\mathrm{NN}_{d}$ and $k\mathrm{NN}_{d'}$ return $k$ nearest neighbors in original and DEA spaces respectively.
  We consider NN coverage as a direct measure of whether the structure of original distance spaces is preserved or not.
  We report NN coverage for $k \in$ \{1, 5, 10, 50, 100, 500, 1,000\} in Section~\ref{sec:exp-dea}.
  The reported values were averaged from 1,000 series, whose kNN was searched from 20,000 series.

  \item \textbf{1st BSF Tightness}. 
  To evaluate the DEA performance on data series similarity search, we used the tightness of the first Best-So-Far (1st BSF)~\cite{c19-vldb-Echihabi-return}.
  In the context of approximate similarity search, 1st BSF is the best result under the constraint of a fixed number of leaf nodes, or series allowed to be examined by the query answering algorithm. 
  In the case where only one leaf node is allowed to be examined, the 1st BSF is also called the \emph{approximate answer}.
  In our experiments (Sections~\ref{sec:exp-design} and~\ref{sec:exp-search}), we report the 1st BSF tightness as a function of the number of series examined (this makes for a fair comparison across indices with leaves containing different number of series).
  Similar to previous work~\cite{c18-vldb-Echihabi-lernaean,c19-vldb-Echihabi-return}, we report the average tightness over 1,000 queries.
\end{enumerate}

\subsection{SoS Preservation and SEAsam}\label{sec:exp-design}

In this section, we evaluate the two novel design choices we propose for DEA methods, i.e., the SoS preservation principle and SEAsam.

\noindent{\bf [SoS Preservation]}\label{sec:exp-sos}
First, we evaluate the effect of the scaling steps introduced by the SoS preservation. 
We trained all five models using SEAsam samples across seven 100M-size datasets and reported their 1st BSF tightness improvements.
1st BSFs were reported under the constraint that the query answering algorithm examines a maximum of 10,000 series in the index before producing the answer.
Improvements are calculated by subtracting the 1st BSF tightness of the non-scaled models from that of the scaled models.

The results show %
that the scaled models provided better 1st BSF in 32 out of the 35 experiments (91\%).
The only three exceptions were by small margins.
Besides, 14 of 35 (40\%) non-scaled models could not effectively converge.
They either converged to bad local optima of large constant DEAs (cf. Section~\ref{sec:ss}), or did not converge and generate random DEAs (exhibiting similar statistics to %
Table~\ref{tab:loss-scaling}).

These results verify that the proposed SoS preservation is indeed an effective method for both preserving pairwise distances and facilitating network convergence.
We note that the SoS preservation idea is applicable to any suitable architecture, and apart from SEAnet, it also improves the performance of the non-scaled versions of FDJNet, TimeNet and InceptionTime.
In the rest of this section, we only report results using the scaled models.

\begin{figure}
  \centering
  \includegraphics[width=0.95\linewidth]{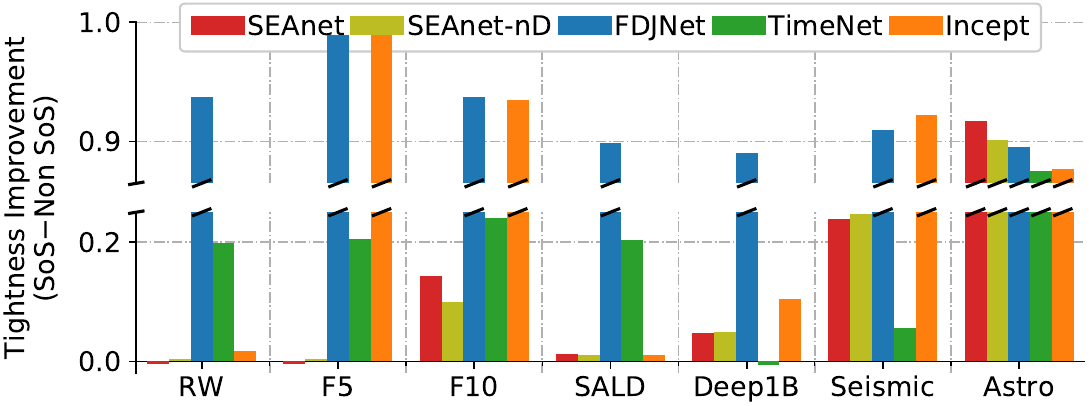}
  \caption{Improvements of employing SoS preservation in terms of 1st BSF tightness (positive values mean SoS is better); 100M datasets.}
  \label{fig:exp-scaling}
\end{figure}

\noindent{\bf [SEAsam]}\label{sec:exp-seasam}
Second, we compare the proposed SEAsam against the commonly used uniform random sampling. %
Results are reported similarly to the previous experiments.
Improvements were calculated by subtracting the 1st BSF tightness of models trained using random samples from those trained using SEAsam samples.

For 27 out of the 35 experiments (77\%), SEAsam provided tighter 1st BSFs than random sampling. %
SEAsam was only surpassed by random sampling on 8 experiments (23\%) with a very small margin. %
We observe that for SEAnet, %
SEAsam was always better. 

\begin{figure}
  \centering
  \includegraphics[width=0.95\linewidth]{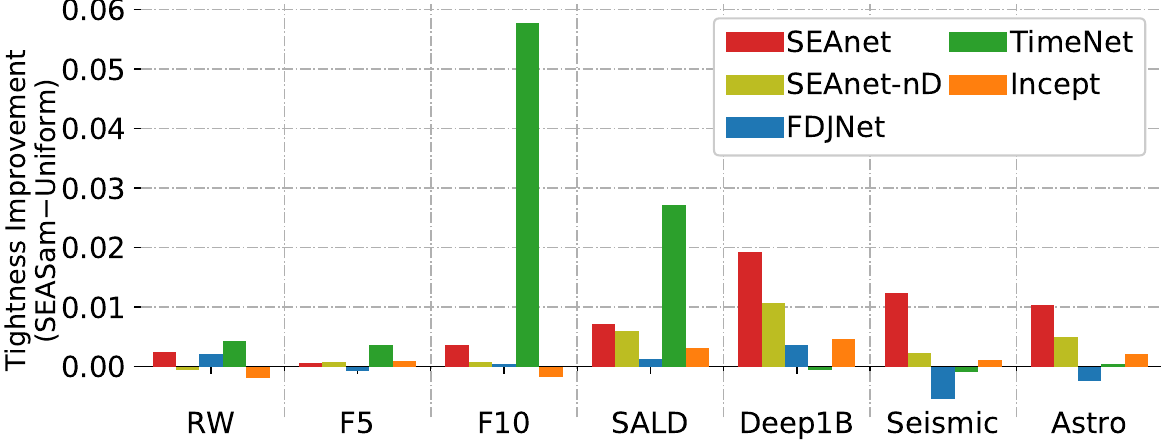}
  \caption{Improvements of SEAsam over uniformly random sampling in terms of 1st BSF rightness
  (positive values mean SEAsam is better); 100M datasets.}
  \label{fig:exp-sampling}
\end{figure}

In order to evaluate the effectiveness of SEAsam, we also measured the number of distinct leaves (of an index constructed on the entire dataset) containing a series that is part of the SEAsam sample in Figure~\ref{fig:exp-sampling-coverages}.
Intuitively, the leaf nodes of the index represent an effective split of the series space, which corresponds to the underlying distribution of the collection.
The more leaf nodes a sample set covers, the better it represents the entire collection.
In our experiments, for all samples with sizes between 10K-500K series across our 7 datasets, SEAsam samples covered more leaf nodes than uniformly random samples with an average improvement of 8\%, and up to 28\% %
for the challenging F10 dataset.

\begin{figure*}[tb]
  \centering
  \subfloat{\includegraphics[width=\textwidth]{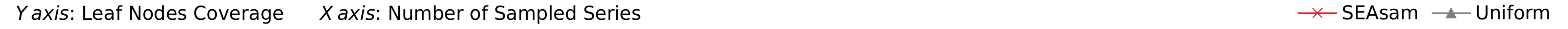}}\\[-2ex]
  \setcounter{subfigure}{0}
  \subfloat[RandWalk]{
    \includegraphics[width=.137\textwidth]{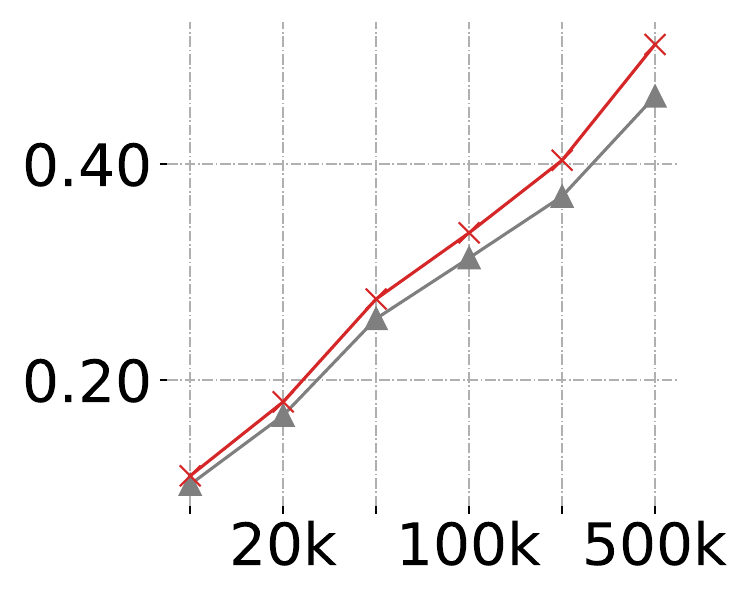}}
  \subfloat[F5]{
    \includegraphics[width=.137\textwidth]{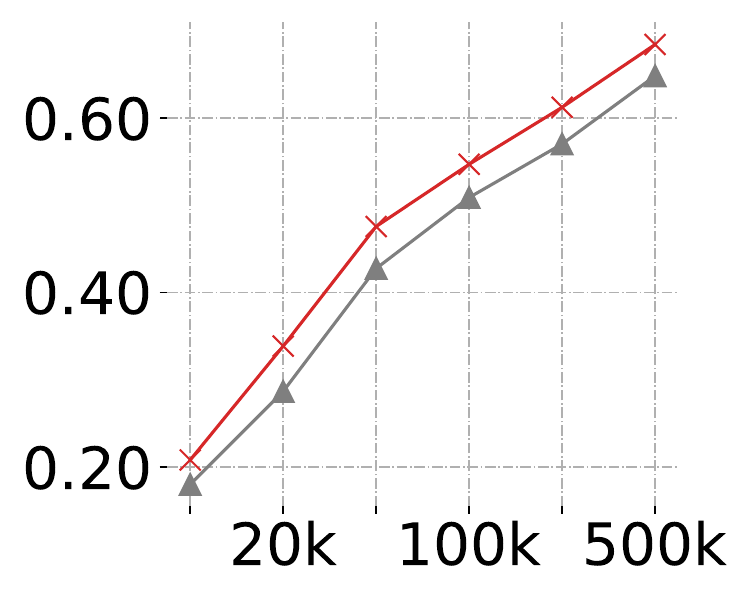}}
  \subfloat[F10]{
    \includegraphics[width=.137\textwidth]{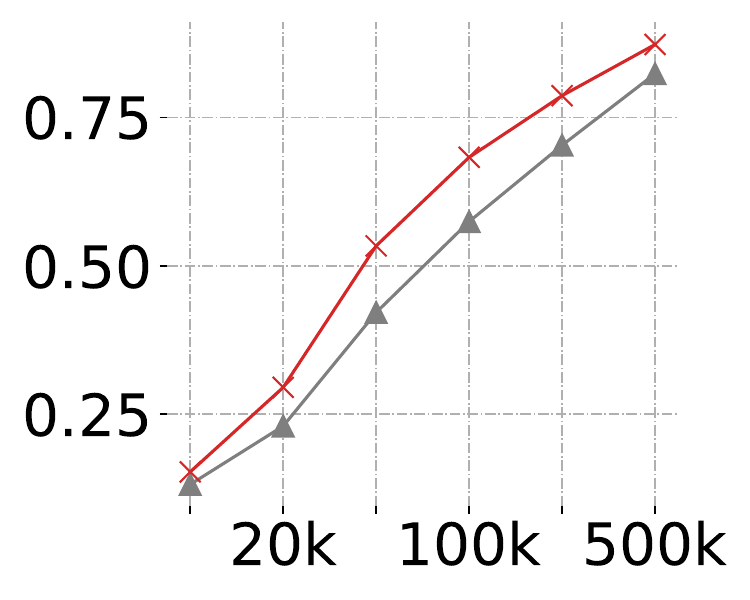}}
  \subfloat[SALD]{
    \includegraphics[width=.137\textwidth]{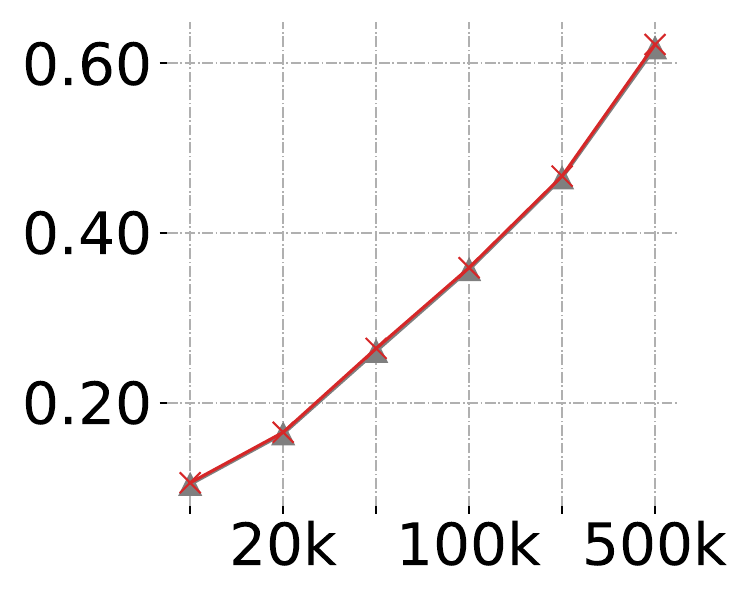}}
  \subfloat[Deep1B]{
    \includegraphics[width=.137\textwidth]{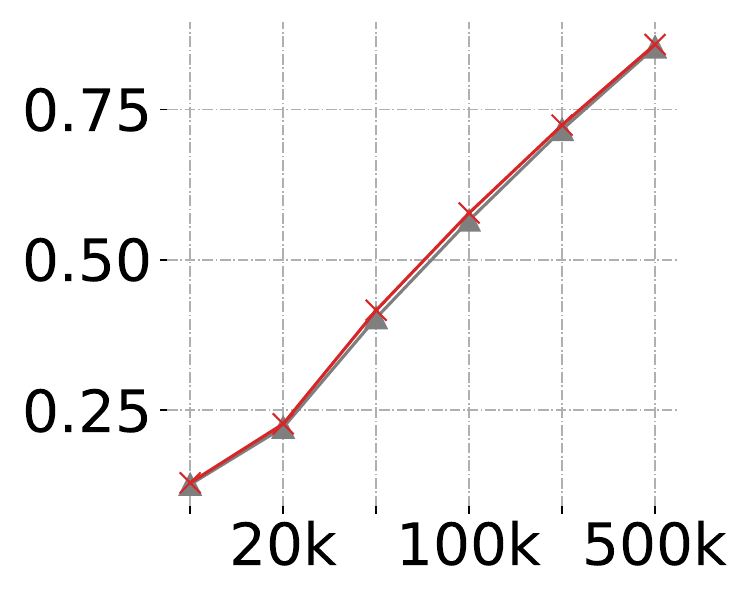}}
  \subfloat[Seismic]{
    \includegraphics[width=.137\textwidth]{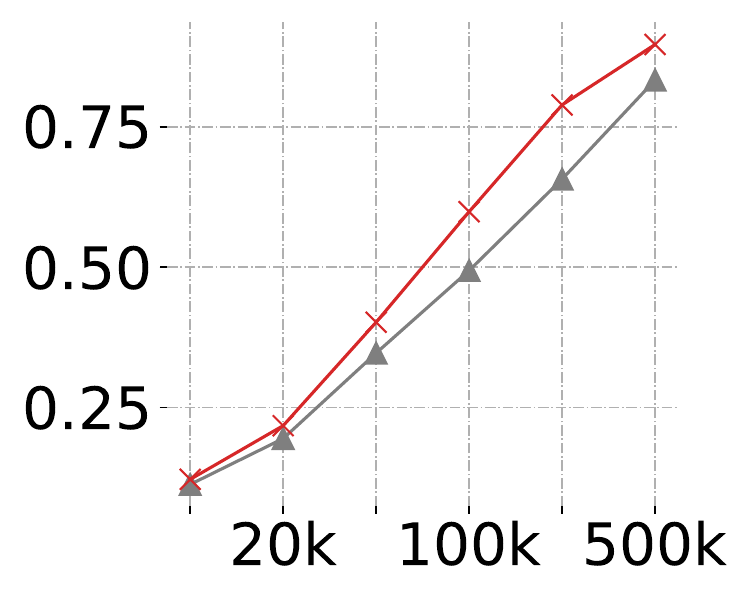}}
  \subfloat[Astro]{
    \includegraphics[width=.137\textwidth]{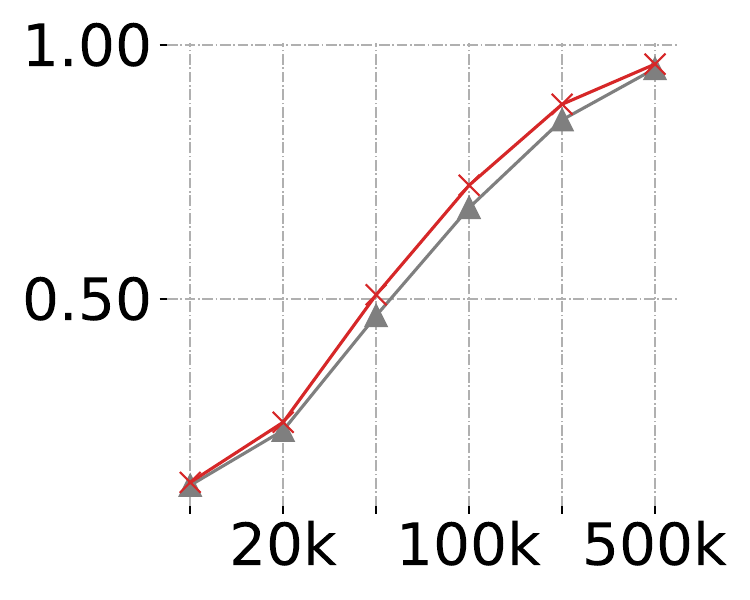}}
  \caption{Leaf node coverage across different sampling sizes for uniformly random sampling (light gray) and SEAsam (dark red).}
  \label{fig:exp-sampling-coverages}
\end{figure*}

These results verify that SEAsam provides more representative samples than uniformly random sampling.
In the following, we report results with SEAsam.

\subsection{DEA Quality}\label{sec:exp-dea}

\begin{table*}[tb]
  \centering
  \caption{(a) Averaged distance differences between pairs of series in the original and the embedded (PAA, DEA) spaces. 
  (b) Root-Mean-Square error between original and reconstructed series. In both cases, best result (lower is better) is marked in bold, second best is underlined. 
  (Results calculated using 10,000 series SEAsam sampled from datasets of 10-million series.)}
  \vspace{-0.5\baselineskip}
  \label{tab:distances}
  {\small
  \setlength\tabcolsep{4pt}
  \begin{tabular}{@{}c|cccccc|cccc@{}} 
    \multicolumn{1}{c}{ } & \multicolumn{6}{c}{$($a$)$ Averaged Distance Differences} & \multicolumn{4}{c}{$($b$)$ Reconstruction RMS Error} \\ 
    \cmidrule(rl){2-7} \cmidrule(rl){8-11}
    \multicolumn{1}{c}{Dataset} & PAA & FDJNet & TimeNet & \multicolumn{1}{c}{Incept} & SEAnet-nD & SEAnet & PAA & TimeNet & Incept & SEAnet \\
    \midrule
    RandWalk & 1.3701 & \underline{0.2794} & 0.4098 & 0.6285 & 0.2976 & \textbf{0.2194} 
    & \underline{0.3061} & 0.3354 & 0.3587 & \textbf{0.2604} \\
    F5 & 2.1152 & 0.1737 & 0.2103 & 0.2836 & \underline{0.1692} & \textbf{0.1629}
    & 0.4214 & \underline{0.2527} & 0.2708 & \textbf{0.2433} \\
    F10 & 5.0395 & 1.1943 & 1.9063 & 1.3958 & \underline{1.1859} & \textbf{1.1672} 
    & 0.6238 & 0.6799 & \underline{0.5041} & \textbf{0.2635} \\
    \midrule
    SALD & 3.2927 & 0.6247 & 0.6928 & 0.858 & \textbf{0.5748} & \underline{0.6182} 
    & \underline{0.5586} & 0.5883 & 0.6831 & \textbf{0.5023} \\
    Deep1B & 8.1095 & 0.9511 & 7.8478 & 0.9511 & \textbf{0.9083} & \underline{0.9484} 
    & 0.9207 & 1.0 & \underline{0.6368} & \textbf{0.5418} \\
    Seismic & 9.9629 & \underline{1.3798} & 1.6577 & 1.4555 & \textbf{1.306} & 1.4514 
    & \textbf{0.7385} & 0.8211 & 0.9669 & \underline{0.7771} \\
    Astro & 14.622 & 1.9239 & 2.4981 & \textbf{1.7983} & \underline{1.8991} & 1.9737 
    & \textbf{0.9267} & \underline{1.0} & 1.4096 & 1.1196 \\  
  \end{tabular}
  } %
\end{table*}

\noindent{\bf [Average Distance Differences]}\label{sec:exp-difference}
The averaged distance differences reported in Table~\ref{tab:distances}a, show that SEAnet and SEAnet-nD outperformed PAA in all 7 datasets.
SEAnet and SEAnet-nD also outperformed all other architectures in 6 out of the 7 datasets. %
The averaged distance differences for SEAnet were better than SEAnet-nD for the 3 synthetic datasets, but worse for the 4 real datasets.
This is because of the regularization effect of the decoder.
Synthetic datasets are less noisy, making models prone to overfitting to the training sets.
In this case, the decoder's regularization improves the averaged distance differences.
On the other hand, real datasets are more noisy, where loss $L_R$ dominates $L_C$, making SEAnet worse than SEAnet-nD in terms of averaged distance differences.
However, as we will explain later on, %
SEAnet still outperformed SEAnet-nD in terms of NN coverages and 1st BSF tightness.

Regarding the other models, TimeNet worked better than FDJNet and InceptionTime only for the F5 dataset, which is of moderate periodicity, and lagged behind for Deep1B, whose adjacent values are less correlated.
InceptionTime's high performance on Astro is an interesting result. 
However, after examining the embedding and reconstructed series of InceptionTime on Astro, we infer this is due to overfitting (reconstruction RMS, NN coverage, 1st BSF tightness and other results concur to this explanation).

\noindent{\bf [Distance Scatter]}\label{sec:exp-scatter}
Distance differences are depicted in scatter plots in Figure~\ref{fig:exp-scatter}.
Points close to the $y=x$ diagonal correspond to series for which the original distances are well preserved in the DEA space.
We observe that scatters of DEAs generated by SEAnet assembled tighter than PAA around the diagonal for all 7 datasets.
Moreover, scatters of DEAs exhibited stronger linearity; %
thus, SEAnet preserved the true nearest neighborhoods better than PAA.

\begin{figure*}[tb]
  \centering
  \subfloat{\includegraphics[width=\textwidth]{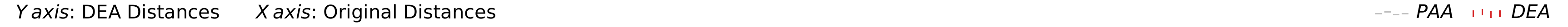}}\\[-2ex]
  \setcounter{subfigure}{0}
  \subfloat[RandWalk]{
    \includegraphics[width=.137\textwidth]{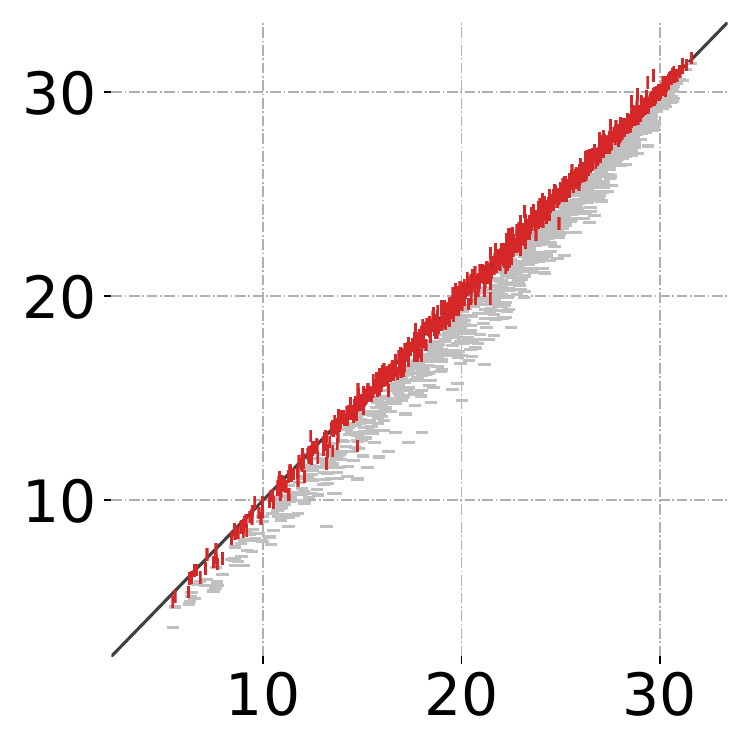}}
  \subfloat[F5]{
    \includegraphics[width=.137\textwidth]{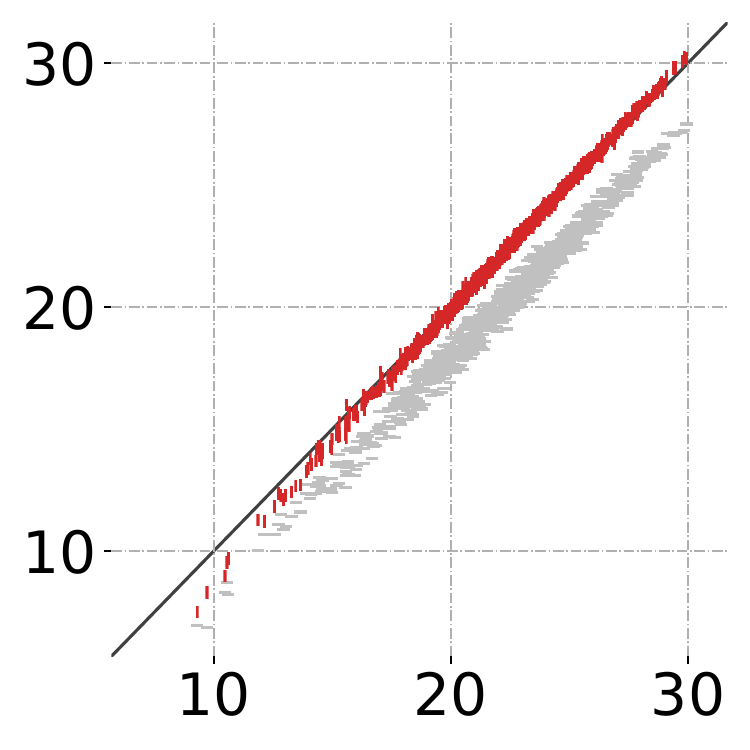}}
  \subfloat[F10]{
    \includegraphics[width=.137\textwidth]{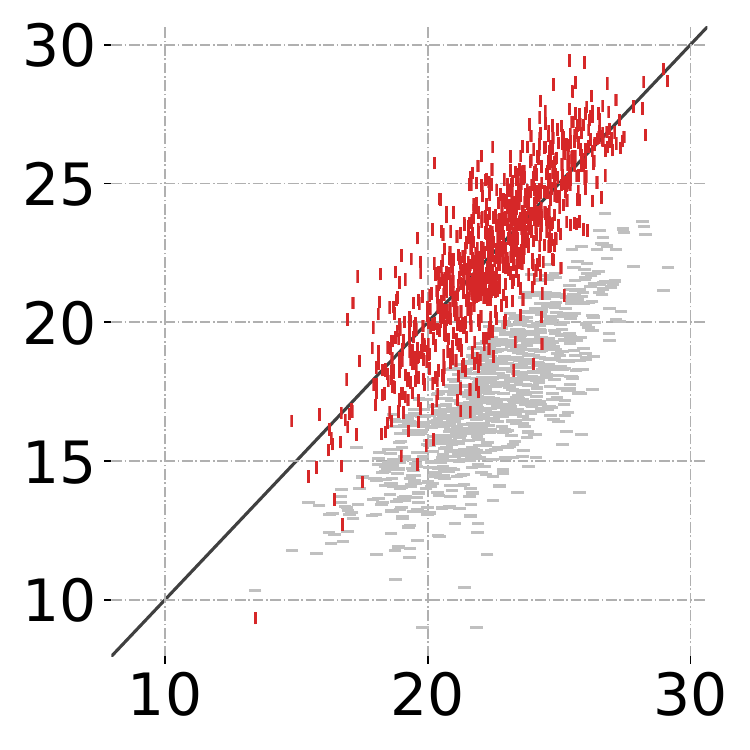}}
  \subfloat[SALD]{
    \includegraphics[width=.137\textwidth]{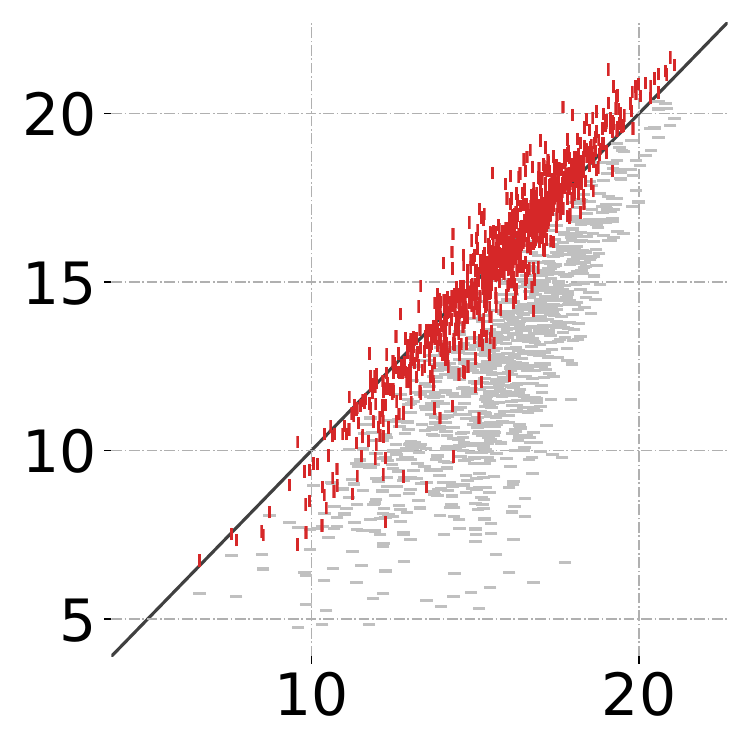}}
  \subfloat[Deep1B]{
    \includegraphics[width=.137\textwidth]{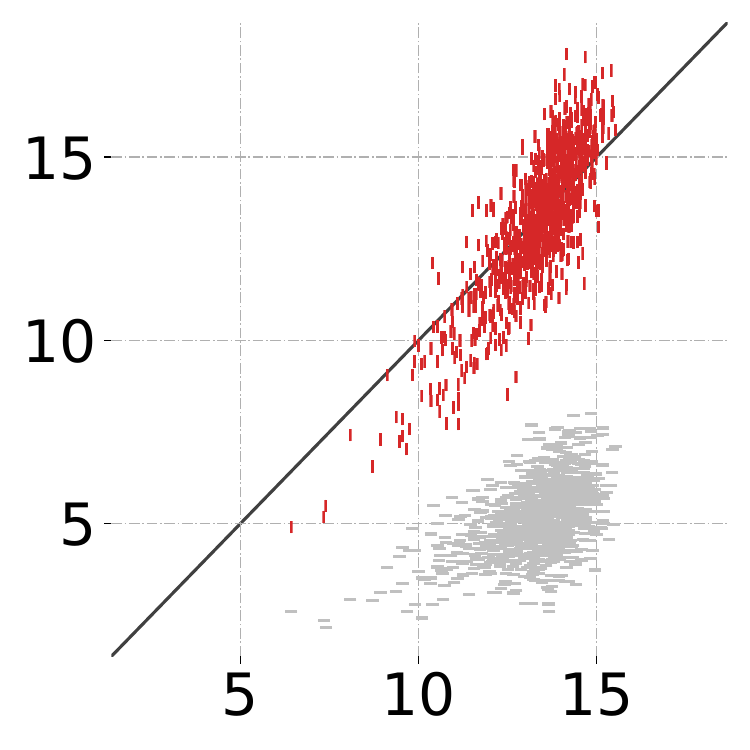}}
  \subfloat[Seismic]{
    \includegraphics[width=.137\textwidth]{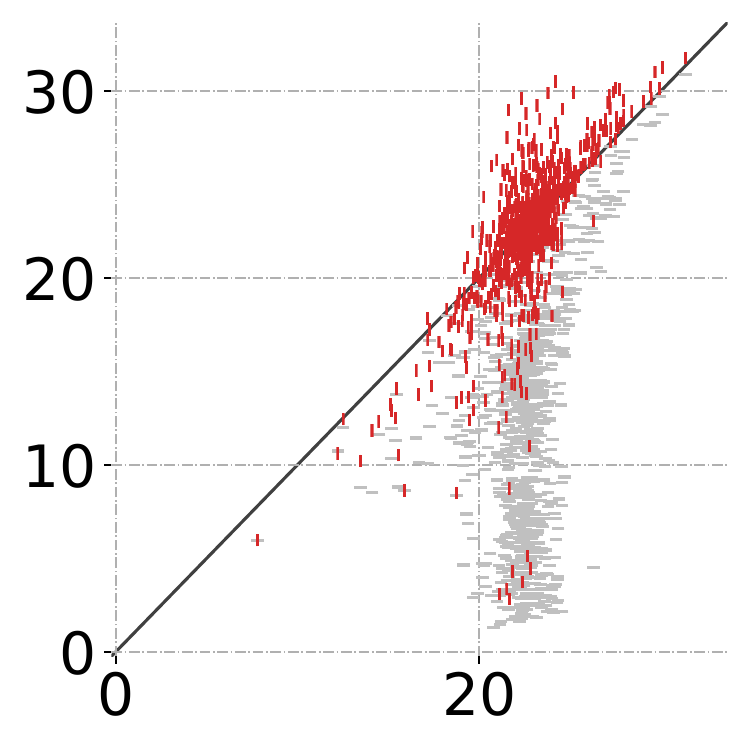}}
  \subfloat[Astro]{
    \includegraphics[width=.137\textwidth]{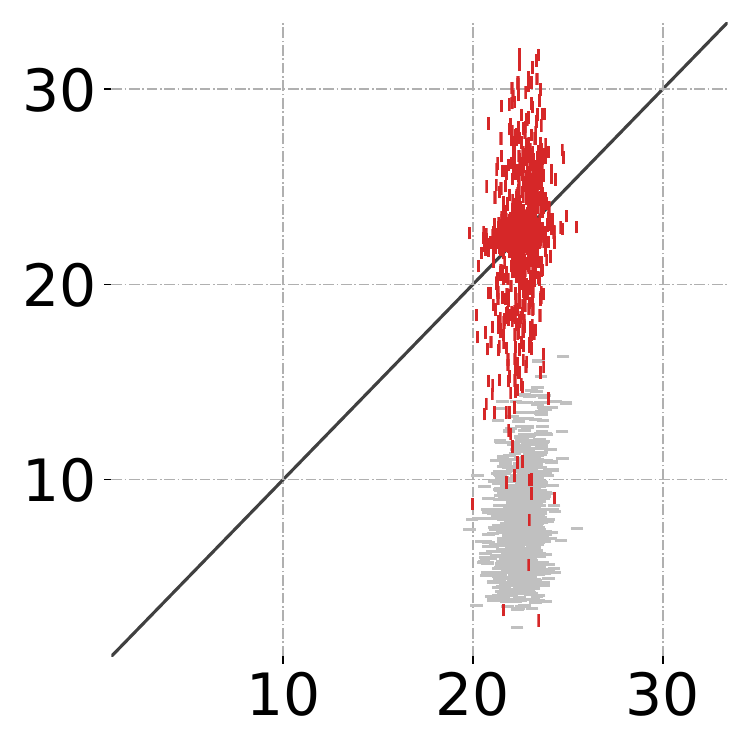}}
  \caption{Distance scatters of 1,000 SEAsam sampled series pairs across different datasets for PAA (light gray) and SEAnet (dark red).}
  \label{fig:exp-scatter}
\end{figure*}

\noindent{\bf [Reconstruction RMS]}\label{sec:exp-rms}
The reconstruction RMS results are reported in Table~\ref{tab:distances}b. 
SEAnet surpassed competitors on 5 out of 7 datasets; it lost to PAA for Seismic and Astro.
Upon close examination, we observe that this happened because of some hard to summarize series, for which neither SEAnet, nor PAA succeeded to produce a good summarization. 

Given that the data series are z-normalized, an RMS $\approx$1 might imply useless local optima reached by setting all reconstruction values to zeroes.
This is exactly the case for TimeNet on Deep1B/Astro.
Such decoders cannot contribute at all to better summarizations.
This is even worse than getting higher RMS,
where the network might still learn from data.
Such observations prove that SEAnet is more robust than PAA, FDJNet, TimeNet and InceptionTime for datasets of different characteristics.

\noindent{\bf [NN Coverages]}\label{sec:exp-nn}
The NN coverages are reported in Figure~\ref{fig:exp-knn}.
SEAnet outperformed PAA and other architectures on all 63 experiments.
This observation confirms SEAnet's capabilities on well-preserving original distance space structures in the DEA space.
The advantage of deep models over PAA is not as obvious as in Table~\ref{tab:distances}a, except for SEAnet, indicating that the target of preserving original distances in DEA distances alone cannot guarantee to provide high-quality DEAs for similarity search.
This observation, together with the fact that SEAnet outperformed SEAnet-nD, confirms the need for the decoder.

\begin{figure*}[tb]
  \centering
  \subfloat{\includegraphics[width=\textwidth]{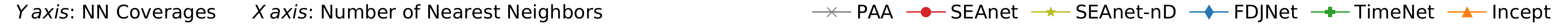}}\\[-2ex]
  \setcounter{subfigure}{0}
  \subfloat[RandWalk]{
    \includegraphics[width=.1385\textwidth]{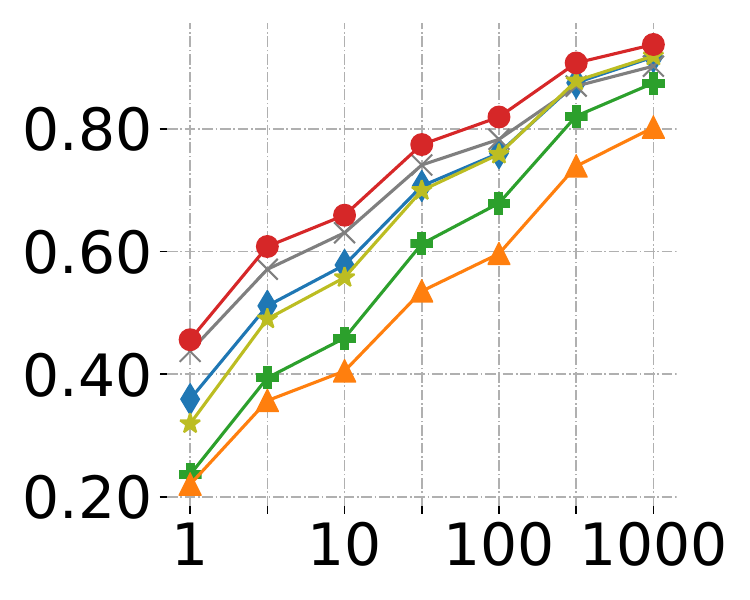}}
  \subfloat[F5]{
    \includegraphics[width=.1385\textwidth]{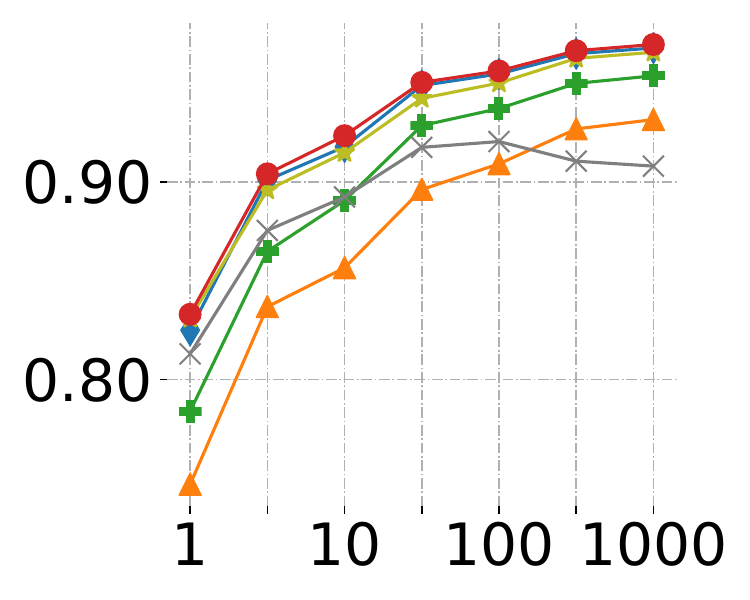}}
  \subfloat[F10]{
    \includegraphics[width=.1385\textwidth]{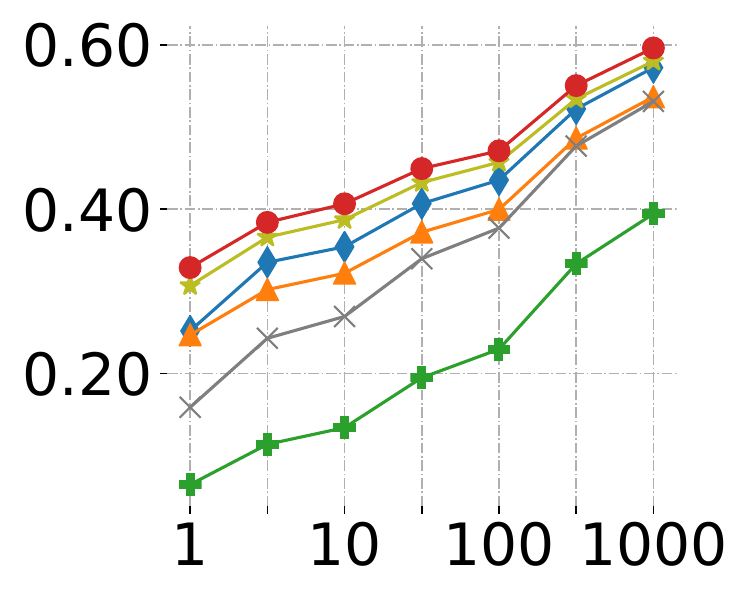}}
  \subfloat[SALD]{
    \includegraphics[width=.1385\textwidth]{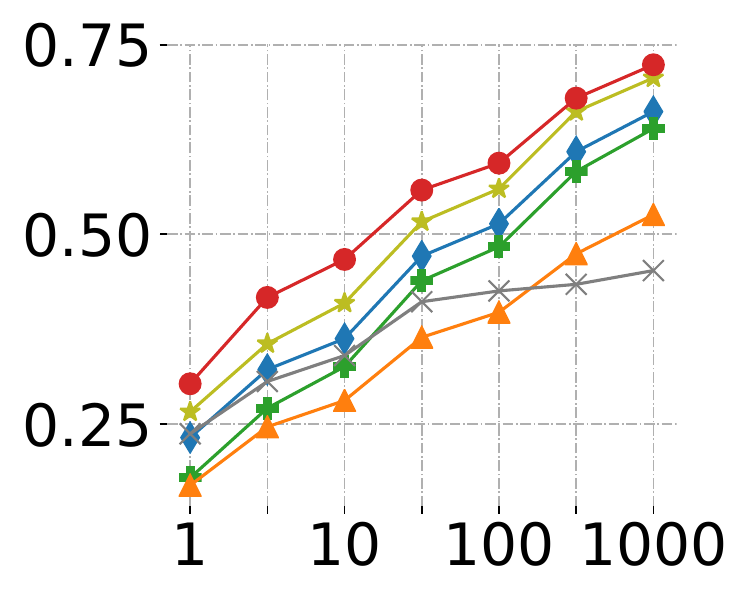}}
  \subfloat[Deep1B]{
    \includegraphics[width=.1385\textwidth]{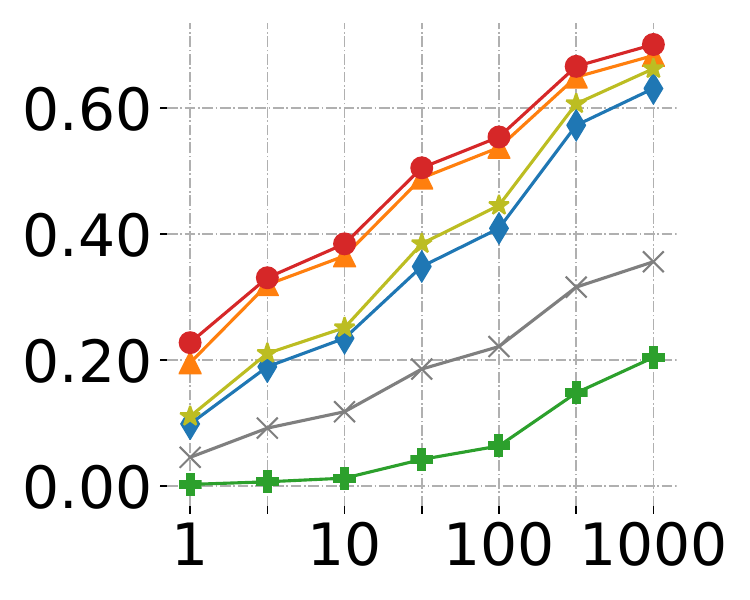}}
  \subfloat[Seismic]{
    \includegraphics[width=.1385\textwidth]{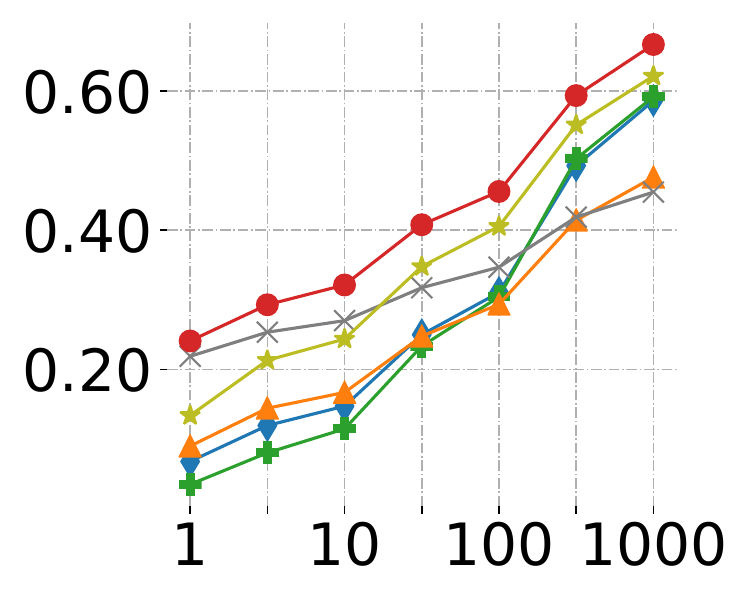}}
  \subfloat[Astro]{
    \includegraphics[width=.1385\textwidth]{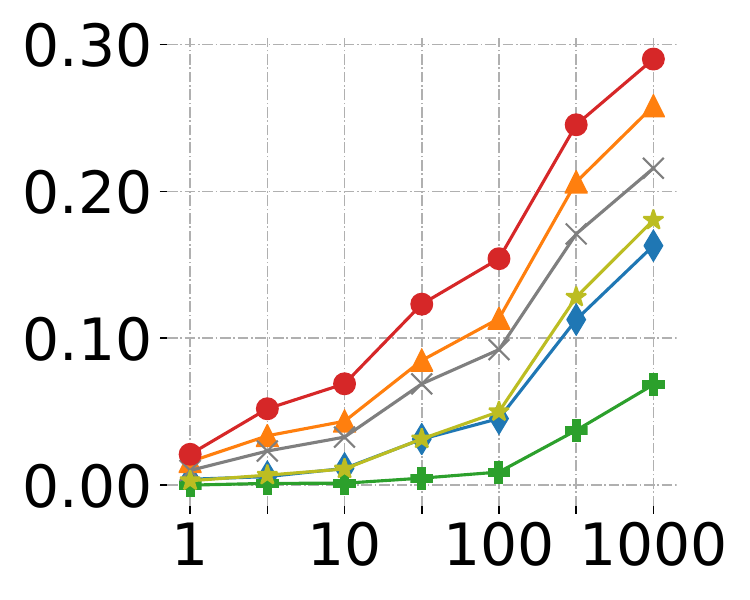}}
  \caption{Nearest neighbors' coverage vs neighborhood size (higher is better).}
  \label{fig:exp-knn}
\end{figure*}

SEAnet-nD outperformed FDJNet on 56 out of the 63 experiments (89\%).
This confirms the overall SEAnet design choices over FDJNet, even after FDJNet was improved by using SoS preservation.
There were clear gaps on Deep1B in Figure~\ref{fig:exp-knn}e between convolutional models and PAA, and between PAA and TimeNet.
This is because Deep1B is from image processing, where adjacent values are not necessarily correlated.
This once again attests to SEAnet's versatility in handling datasets with different properties.

\subsection{DEA for Approximate Search}
\label{sec:exp-search}

\noindent{\bf [Approximate Answers Tightness]}
Figure~\ref{fig:exp-db-sizes} reports the results for the evaluation of the quality of the approximate similarity search answers. 
That is, we evaluate the quality of the 1st BSF answers when the search algorithm examines a single leaf node of the index.
We note that although this kind of evaluation is used in the literature~\cite{c08-kdd-Shieh-isax, c19-vldb-Echihabi-return}, it is not a fair comparison, as the visited leaf sizes may vary across different experiments, datasets, and indexes, depending on the distribution properties of the DEA values.
If the DEA distribution is very skewed (e.g., with many DEAs having the same or similar values), then they tend to gather into a small subset of the leaf nodes, resulting to a few large leaf nodes.
This has two consequences: 

\begin{figure*}[tb]
  \centering
  \subfloat{\includegraphics[width=\textwidth]{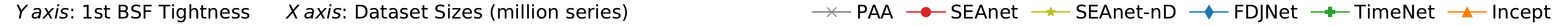}}\\[-2ex]
  \setcounter{subfigure}{0}
  \subfloat[RandWalk]{
    \includegraphics[width=.1385\textwidth]{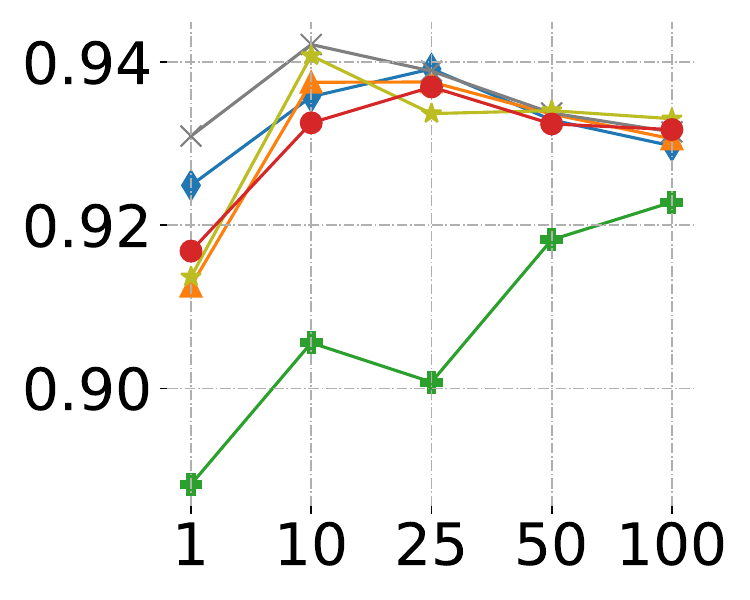}}
  \subfloat[F5]{
    \includegraphics[width=.1385\textwidth]{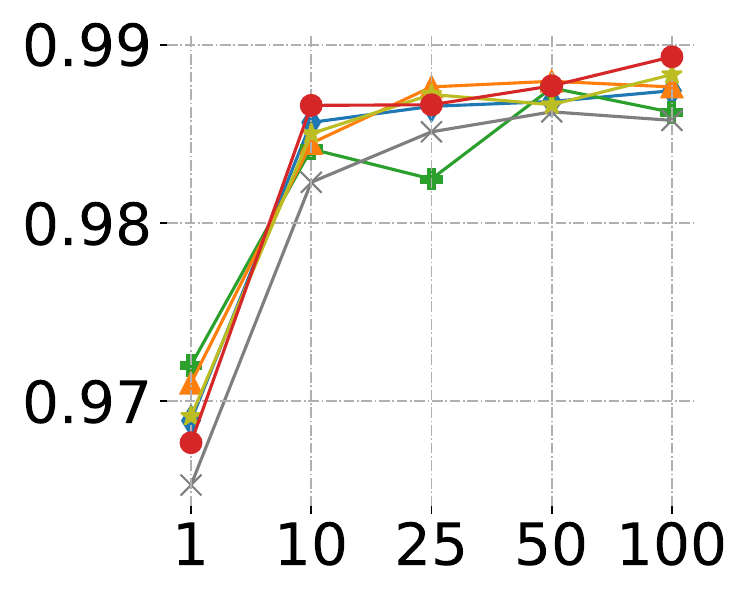}}
  \subfloat[F10]{
    \includegraphics[width=.1385\textwidth]{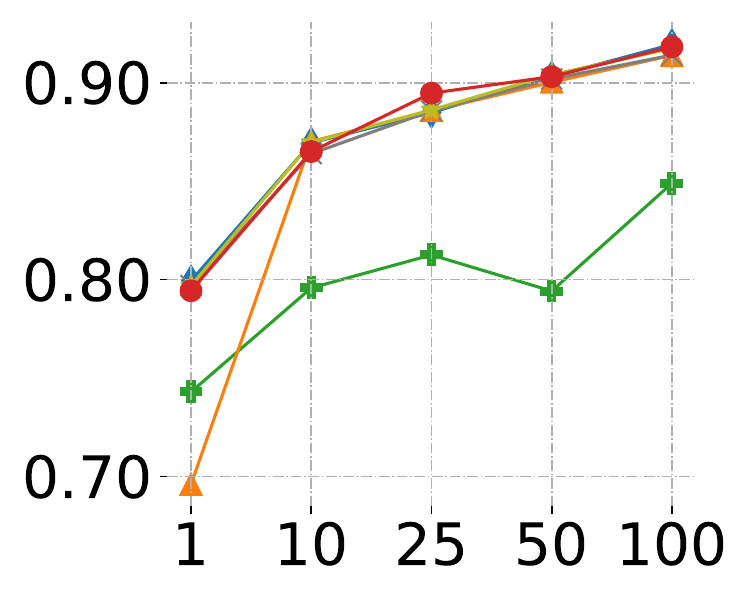}}
  \subfloat[SALD]{
    \includegraphics[width=.1385\textwidth]{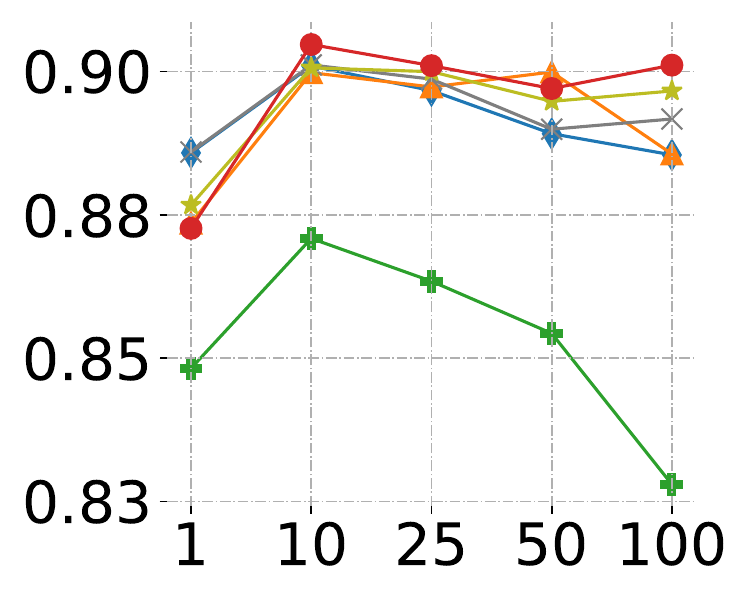}}
  \subfloat[Deep1B]{
    \includegraphics[width=.1385\textwidth]{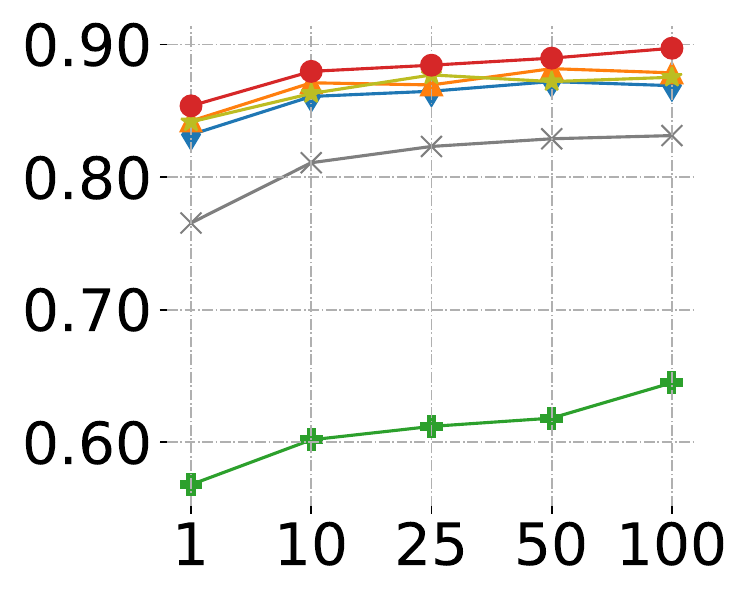}}
  \subfloat[Seismic]{
    \includegraphics[width=.1385\textwidth]{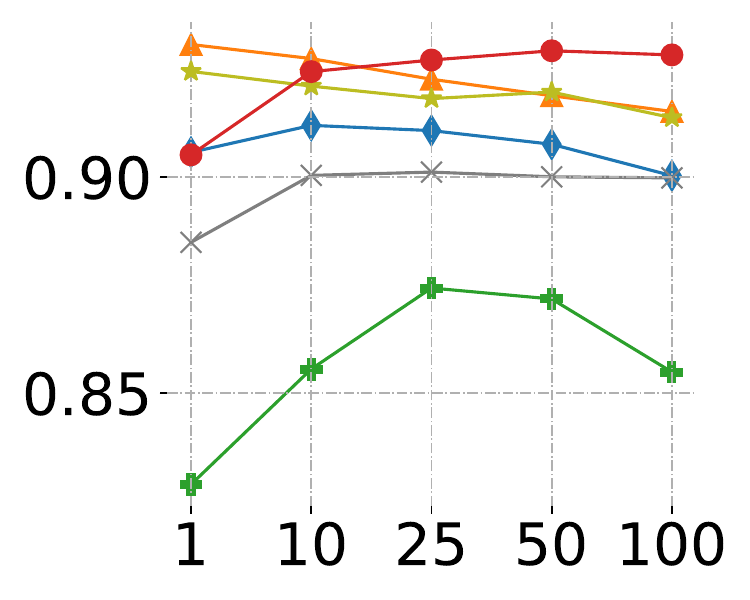}}
  \subfloat[Astro]{
    \includegraphics[width=.1385\textwidth]{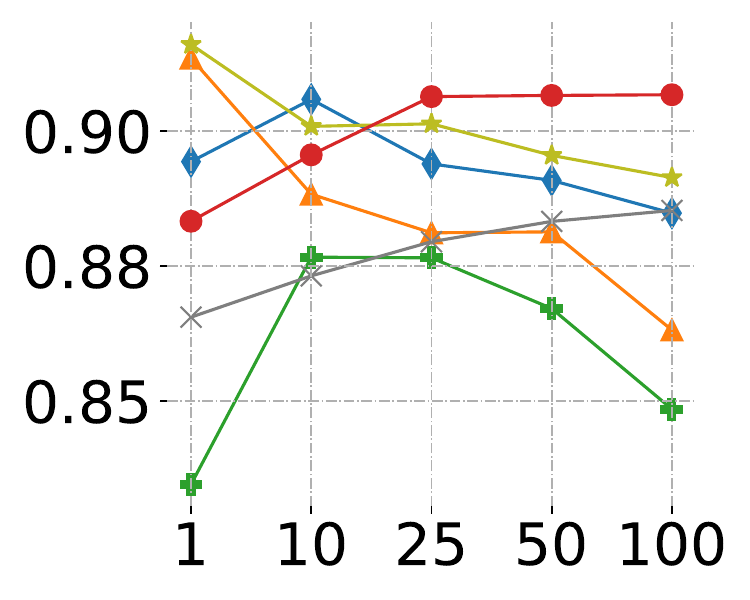}}
  \caption{Approximate query answers quality: 1st BSF tightness vs dataset size (higher is better)}
  \label{fig:exp-db-sizes}
\end{figure*}

\begin{enumerate}[noitemsep, topsep=0pt, wide, labelwidth=!, labelindent=0pt]
  \item Visiting a single large leaf involves examining a large number of series, which implies an increased probability for finding a tighter 1st BSF. 
  Nevertheless, this problem is alleviated as the dataset size increases (e.g., when we move from the 1M to the 100M RandWalk dataset, the additional number of series visited by PAA-based iSAX, when compared to SEAnet DEA-based iSAX, drops from 205\% to 7\%). %
  \item If the produced DEAs are very similar to one another for a large dataset, then the index cannot even be built.
  This happens when a leaf node has exhaustively used all iSAX bits and cannot be split any further in order to accommodate the incoming series.
  This situation generally occurred for DEAs originating from the non-scaled models, and for PAA on the 100M SALD dataset.
  SEAnet avoids this problem across all our experiments.
\end{enumerate}

Benefiting from leaf nodes of larger sizes, PAA or other models outperform SEAnet in some of the 1M and 10M dataset experiments, even though their summarization qualities are not better (refer to Table~\ref{tab:distances}a and Figure~\ref{fig:exp-knn}).
SEAnet and SEAnet-nD show their superiority in the large ($\geq$25M) and hard (real) datasets.
TimeNet generally lags behind the other solutions, indicating that a direct application of recurrent models is less competitive than convolutional models.
This could come from the fact that RNN models are designed to capture temporal correlations while being relatively less powerful at extracting local patterns~\cite{DBLP:journals/corr/abs-1803-01271}.

\noindent{\bf [Leaf Node Compactness]}
To better understand the quality of DEA-based iSAX, we examine the compactness, i.e., the average distances among series of the visited leaf nodes in Figure~\ref{fig:exp-compactness}.
Smaller average distance indicates the index is better in terms of mapping similar series into the same leaf nodes.
In RandWalk and F5, the advantages of SEAnet are not obvious as all methods perform well.
However, when considering the harder datasets, SEAnet achieves better results in 22 out of the 25 experiments. %
These results demonstrate that SEAnet produces more compact indexes than the competitors, especially on datasets considered hard for similarity search.

\begin{figure*}[tb]
  \centering
  \subfloat{\includegraphics[width=\textwidth]{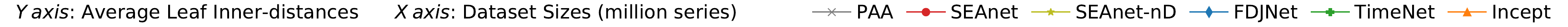}}\\[-2ex]
  \setcounter{subfigure}{0}
  \subfloat[RandWalk]{
    \includegraphics[width=.1385\textwidth]{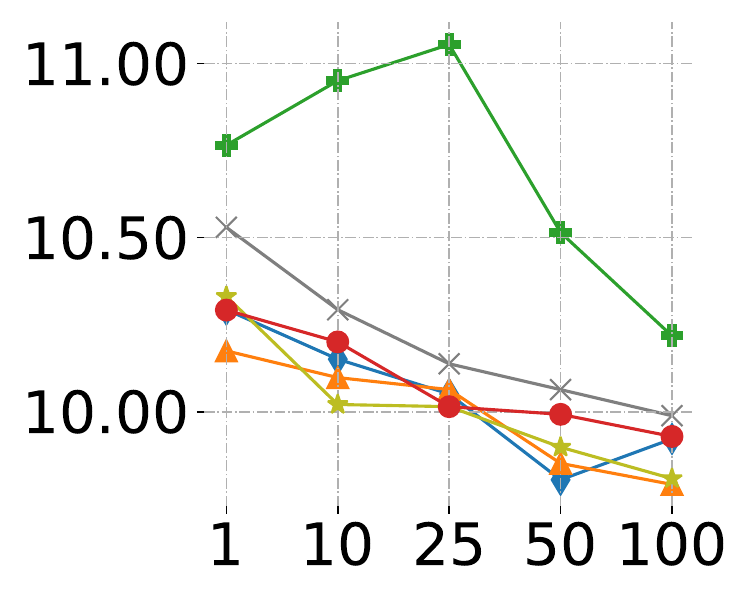}}
  \subfloat[F5]{
    \includegraphics[width=.1385\textwidth]{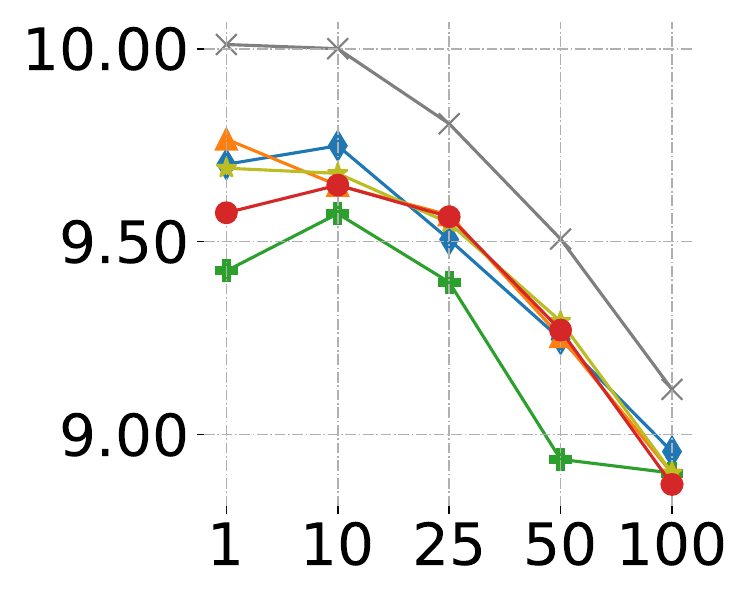}}
  \subfloat[F10]{
    \includegraphics[width=.1385\textwidth]{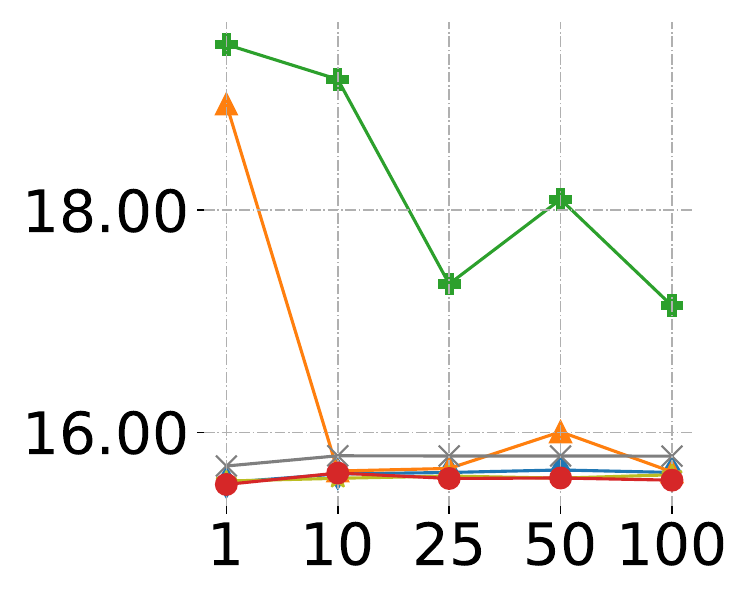}}
  \subfloat[SALD]{
    \includegraphics[width=.1385\textwidth]{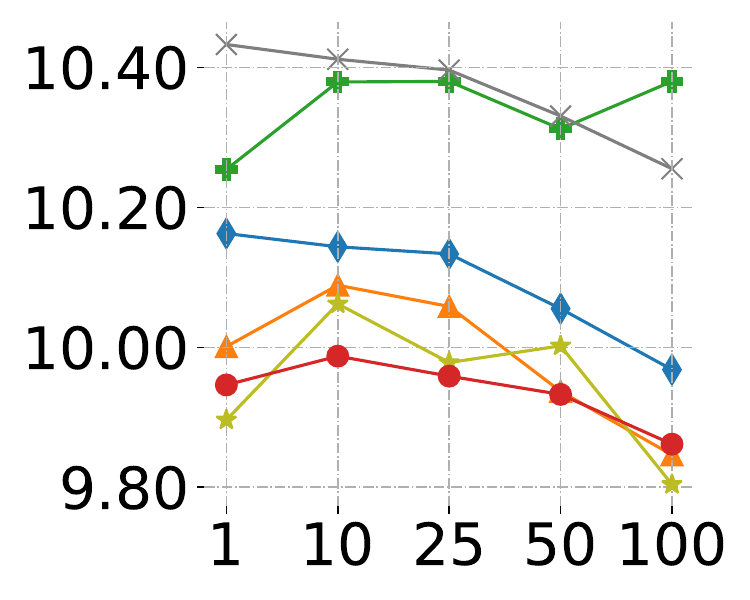}}
  \subfloat[Deep1B]{
    \includegraphics[width=.1385\textwidth]{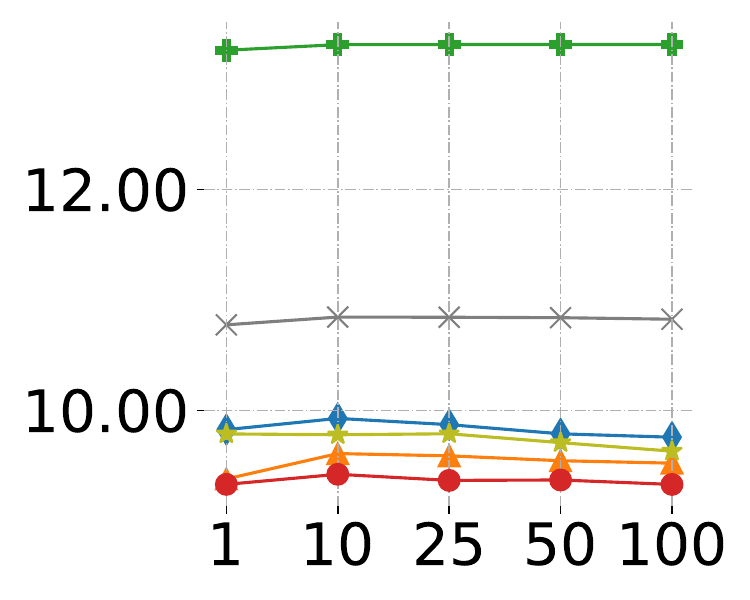}}
  \subfloat[Seismic]{
    \includegraphics[width=.1385\textwidth]{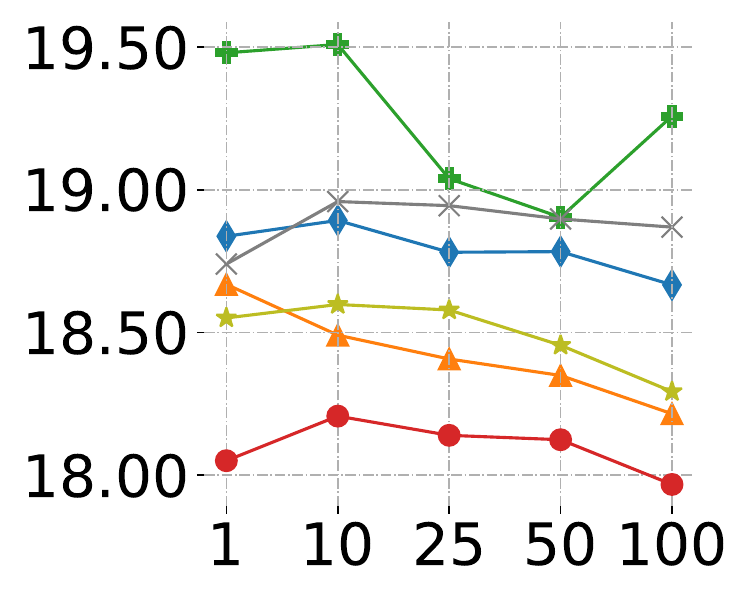}}
  \subfloat[Astro]{
    \includegraphics[width=.1385\textwidth]{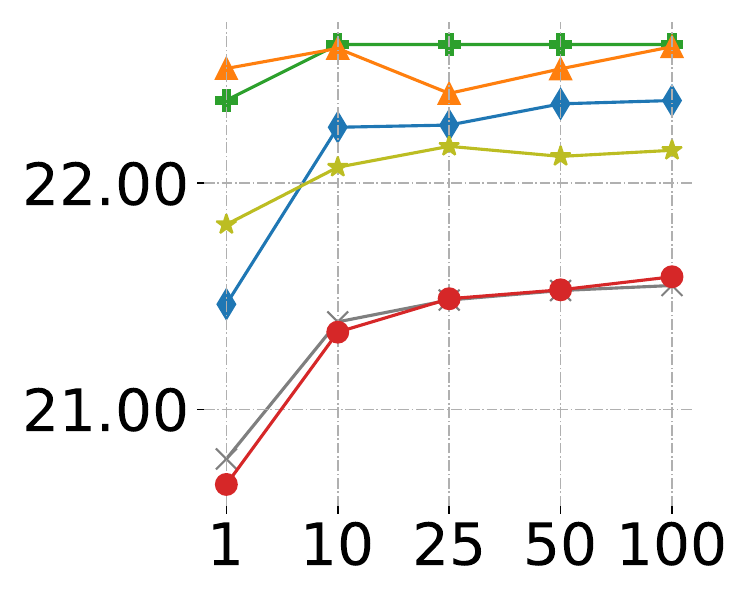}}
  \caption{Average distances among series of the visited leaf nodes vs dataset size (lower is better)}
  \label{fig:exp-compactness}
\end{figure*}

\begin{figure*}[tb]
  \centering
  \subfloat{\includegraphics[width=\textwidth]{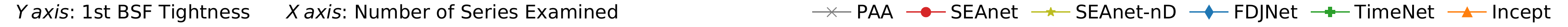}}\\[-2ex]
  \setcounter{subfigure}{0}
  \subfloat[RandWalk\label{fig:exp-max-series-rw}]{
    \includegraphics[width=.1385\textwidth]{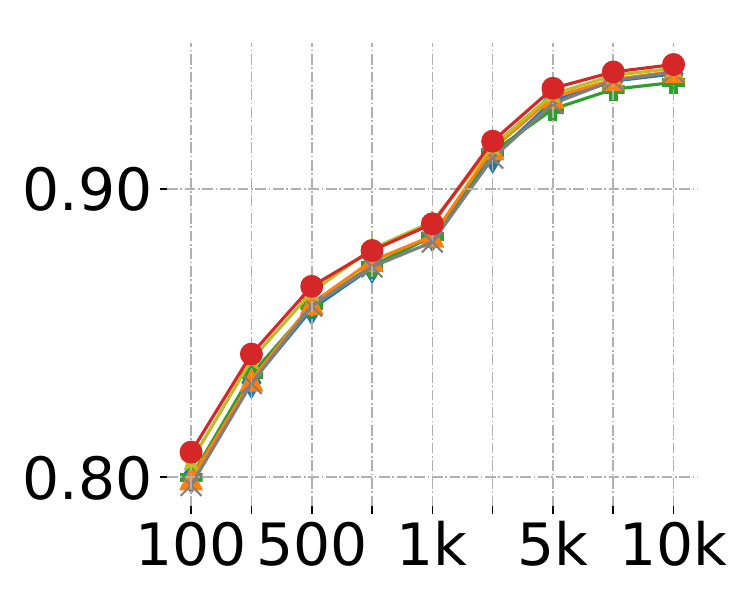}}
  \subfloat[F5\label{fig:exp-max-series-f5}]{
    \includegraphics[width=.1385\textwidth]{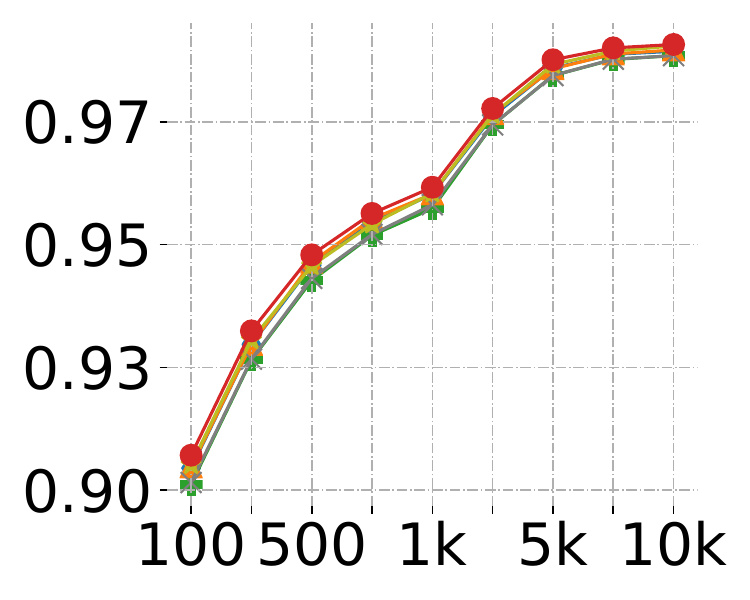}}
  \subfloat[F10\label{fig:exp-max-series-f10}]{
    \includegraphics[width=.1385\textwidth]{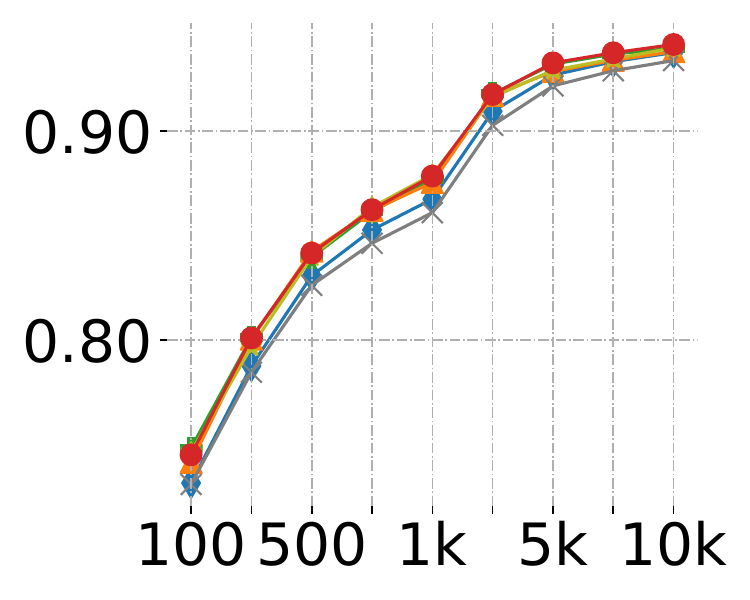}}
  \subfloat[SALD\label{fig:exp-max-series-sald}]{
    \includegraphics[width=.1385\textwidth]{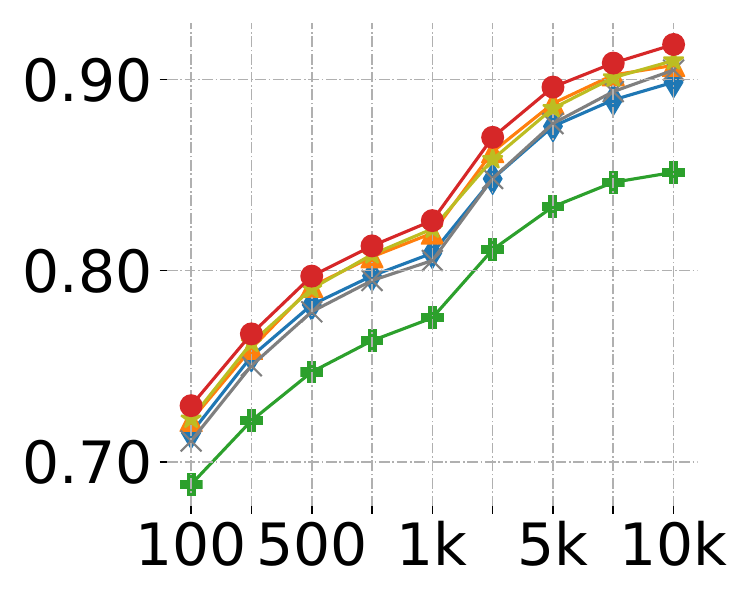}}
  \subfloat[Deep1B\label{fig:exp-max-series-d1b}]{
    \includegraphics[width=.1385\textwidth]{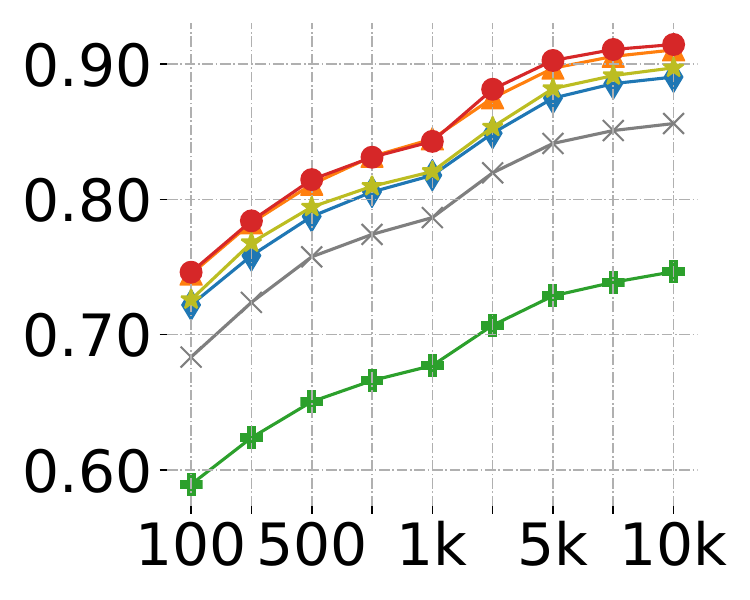}}
  \subfloat[Seismic\label{fig:exp-max-series-sm}]{
    \includegraphics[width=.1385\textwidth]{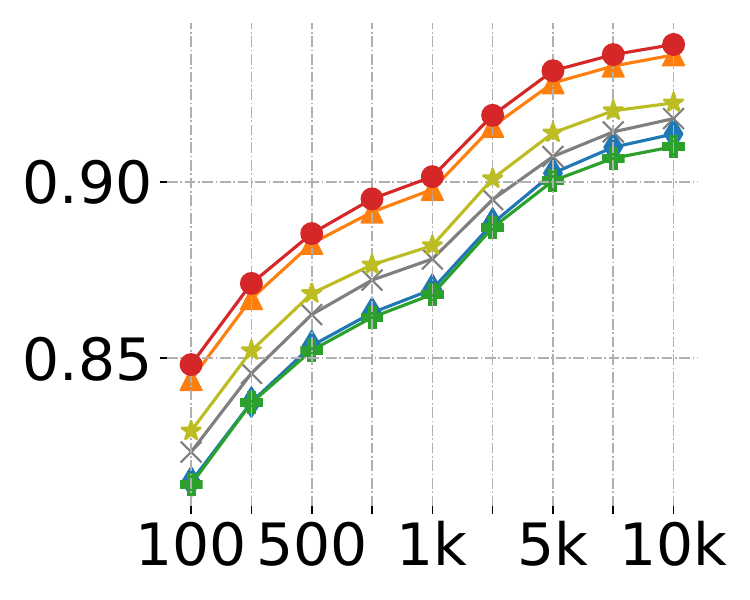}}
  \subfloat[Astro\label{fig:exp-max-series-astro}]{
    \includegraphics[width=.1385\textwidth]{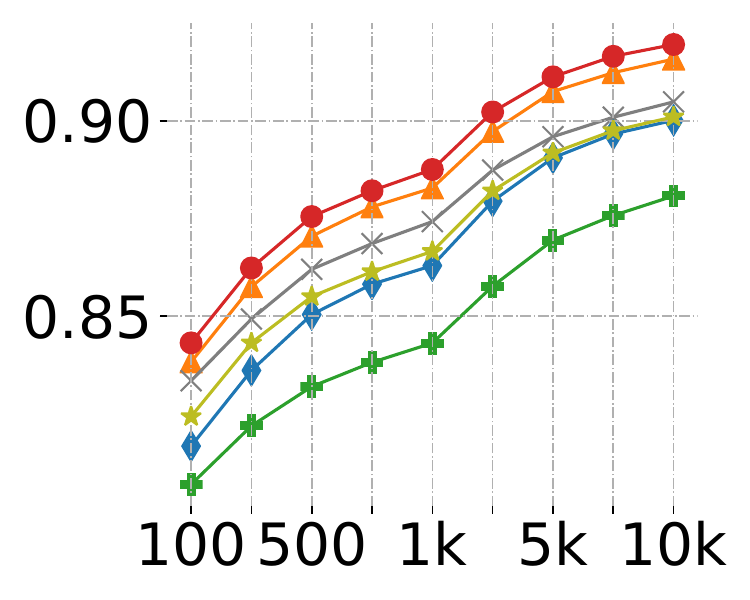}}
  \caption{Approximate query answers quality: 1st BSF tightness vs number of series visited; 100M series datasets (higher is better)}
  \label{fig:exp-max-series}
\end{figure*}

\noindent{\bf [1st BSF Tightness Limited by Number of Series to Examine]}\label{sec:exp-tightness}
We evaluate the benefit of using DEA for similarity search, by reporting the 1st BSF tightness as a function of the number of series that the similarity search algorithm examines. 
The results on 100M datasets, are shown in Figure~\ref{fig:exp-max-series}.
SEAnet improved the 1st BSF tightness, and thus the similarity search results, in 61 out of the 63 experiments. %
Its advantage was particularly obvious on the Deep1B, Seismic, and Astro (hard) datasets.

Besides the 1st BSF tightness, the index's leaf node compactness (i.e., the average distances among the series in a leaf node) also profiles the quality of the index built on DEA or PAA.
Smaller average distance indicates the index is more successful in grouping similar series into the same leaf node.
SEAnet leaded to an average improvement over PAA in leaf node compactness of 4\%, and up to 14\% %
for the challenging real dataset Deep1B, demonstrating its effectiveness in producing more compact indexes than the SOTA competitors.

To conclude this section, our comprehensive experiments verify the effectiveness of SEAnet's ability to provide better DEAs to facilitate data series similarity search.

\subsection{DEA for Downstream applications}
\label{sec:exp-ucr}

\noindent{\bf [SEAnet for Classification]}
We now evaluate the utility of SEAnet DEA in other (than similarity search) downstream applications.
Figure~\ref{fig:exp-ucr} reports the DEA/PAA-based $k$-NN classification accuracy on the UCR time series classification archive~\cite{DBLP:journals/ieeejas/DauBKYZGRK19} (we used the 75 datasets with $\geq$100 equal-length series and no missing values).
With $k$=1, 2, 3, SEAnet DEAs outperformed PAA on 65\%, 68\% and 69\% of the datasets respectively, demonstrating the superiority of SEAnet DEA over PAA for classification tasks, as well. 

\begin{figure}[tb]
  \centering
  \subfloat{\includegraphics[width=0.48\textwidth]{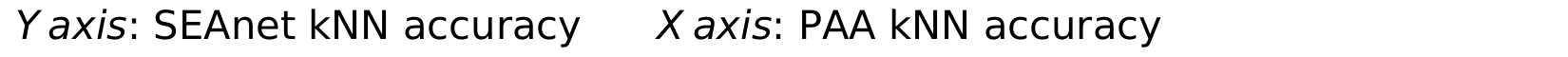}}\\[-2ex]
  \setcounter{subfigure}{0}
  \subfloat[k=1 (49-1-25)]{
    \includegraphics[width=.16\textwidth]{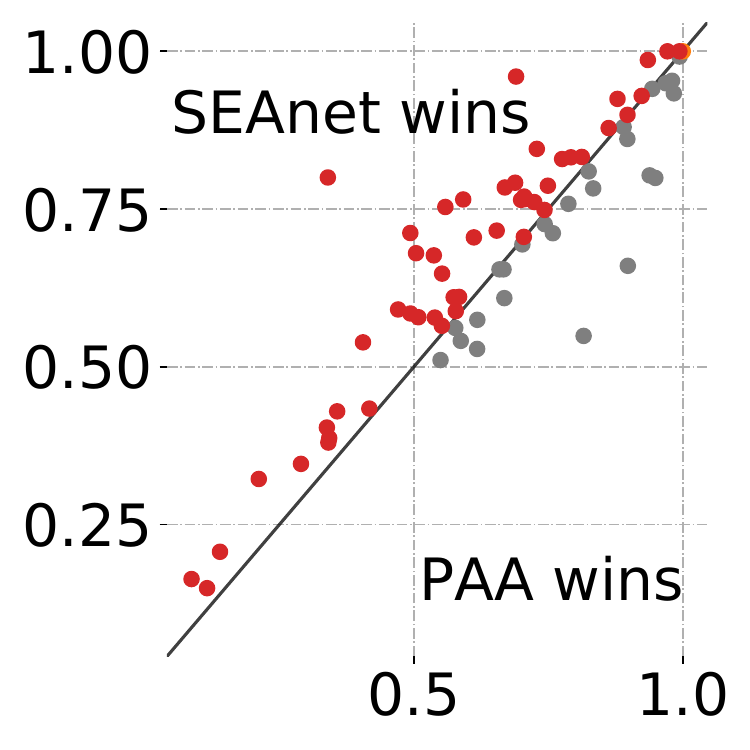}}
  \subfloat[k=3 (51-3-21)]{
    \includegraphics[width=.16\textwidth]{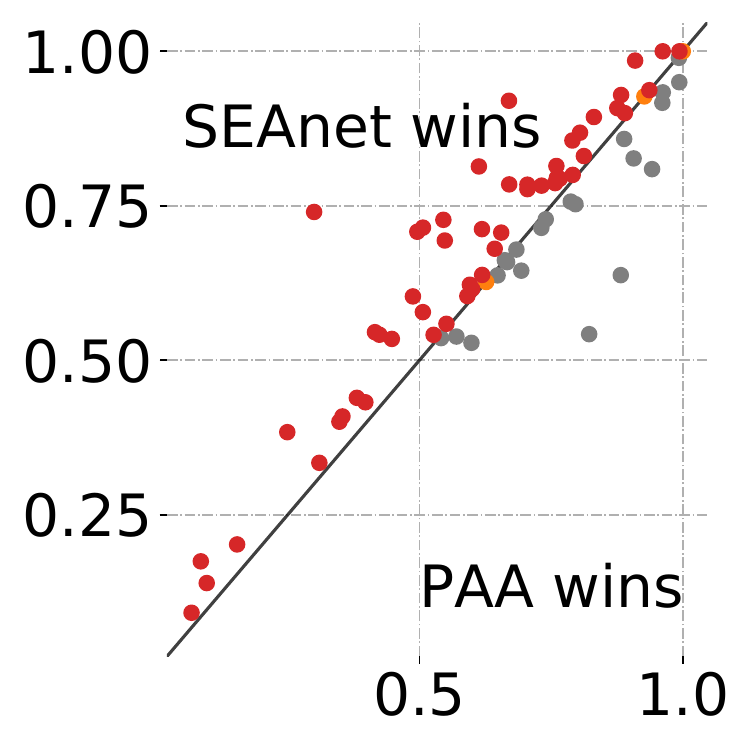}}
  \subfloat[k=5 (52-3-20)]{
    \includegraphics[width=.16\textwidth]{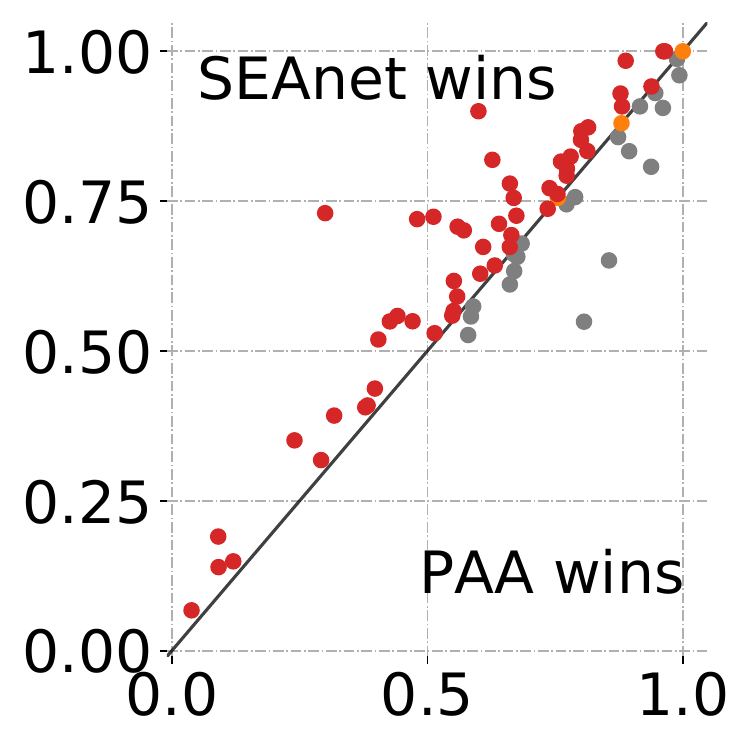}}
  \caption{SEAnet DEA vs. PAA on kNN classification accuracies of UCR18 dataset. Number of datasets where SEAnet wins-ties-losses compared with PAA are reported.}
  \label{fig:exp-ucr}
\end{figure}

\begin{figure}[tb]
  \centering
  \subfloat[Convergences of $L_C$\label{fig:exp-sub-converge}]{
    \includegraphics[width=0.95\linewidth]{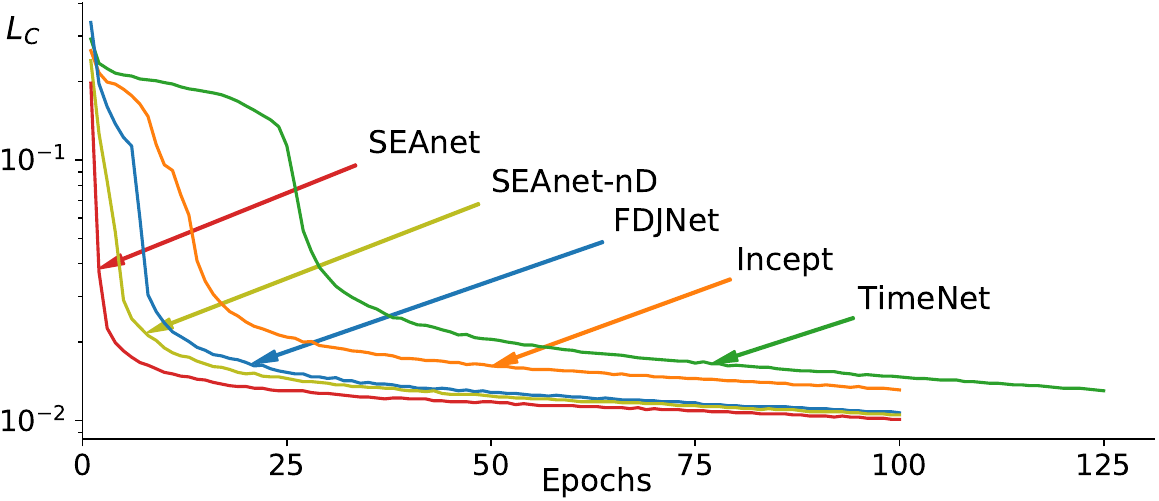}}
  \hfill
  \subfloat[Training speed\label{fig:exp-sub-time}]{
    \includegraphics[width=0.97\linewidth]{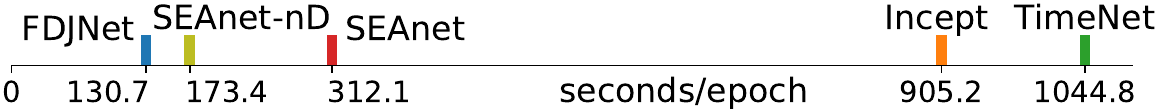}}
  \caption{Convergence procedures measured by compression error (averaged distance difference) $L_C$ across architectures on F5 dataset.}
  \label{fig:exp-convergence}
\end{figure}

\begin{figure}
  \centering
  \includegraphics[width=0.95\linewidth]{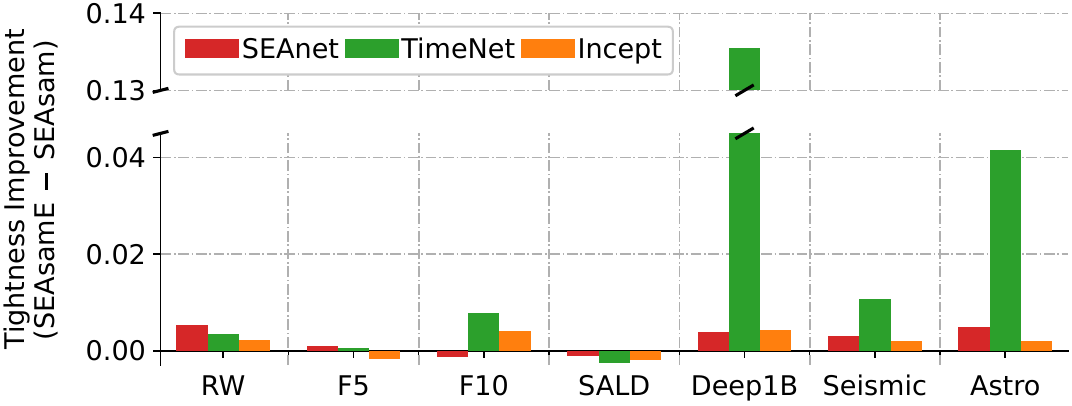}
  \caption{Improvements of SEAsamE over SEAsam in terms of 1st BSF rightness
  (positive values mean SEAsamE is better); 100M datasets.}
  \label{fig:exp-seasame}
\end{figure}

\begin{figure*}[tb]
  \centering
  \subfloat{\includegraphics[width=\textwidth]{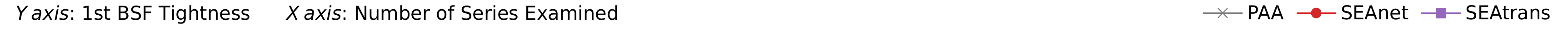}}\\[-2ex]
  \setcounter{subfigure}{0}
  \subfloat[RandWalk\label{fig:re-exp-max-series-rw}]{
    \includegraphics[width=.1385\textwidth]{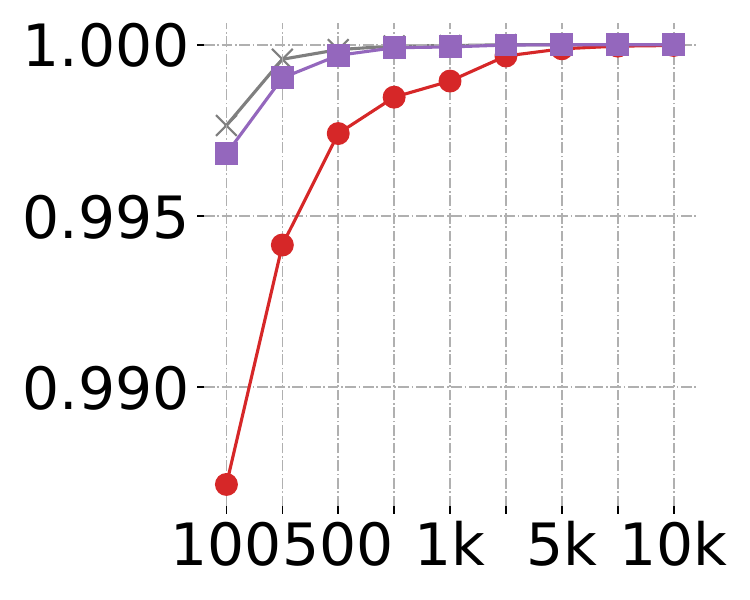}}
  \subfloat[F5\label{fig:re-exp-max-series-f5}]{
    \includegraphics[width=.1385\textwidth]{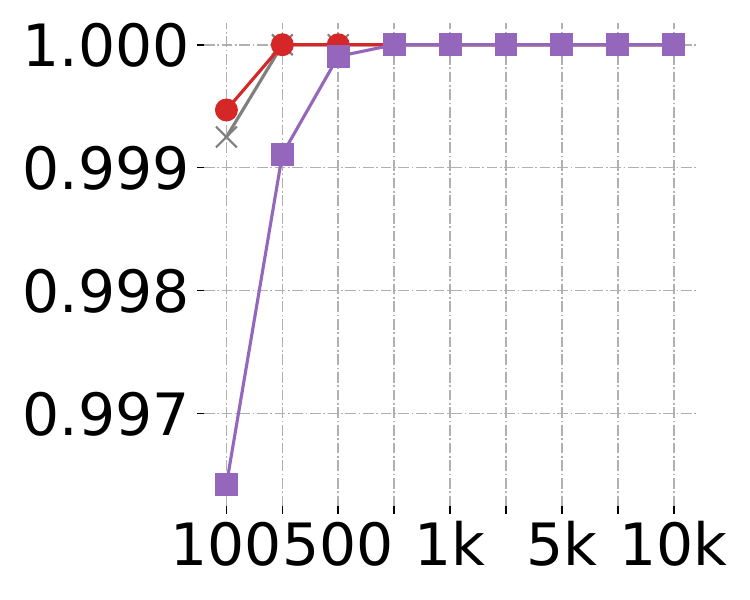}}
  \subfloat[F10\label{fig:re-exp-max-series-f10}]{
    \includegraphics[width=.1385\textwidth]{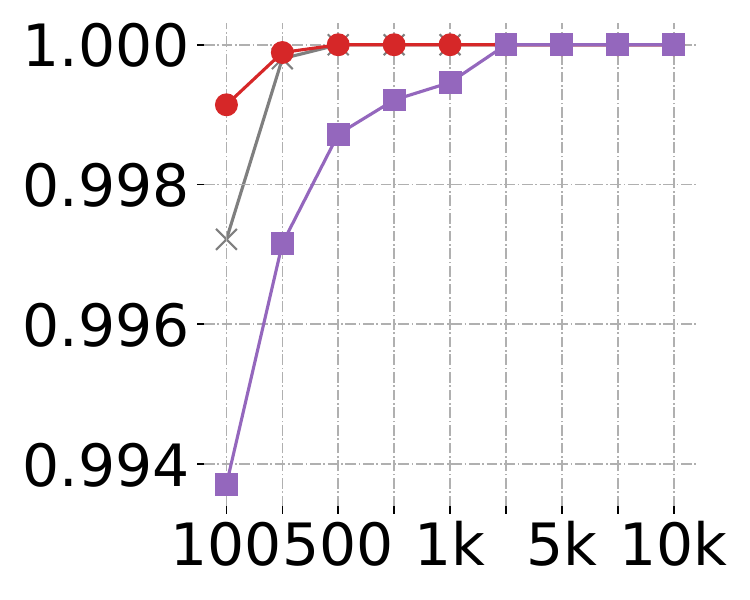}}
  \subfloat[SALD\label{fig:re-exp-max-series-sald}]{
    \includegraphics[width=.1385\textwidth]{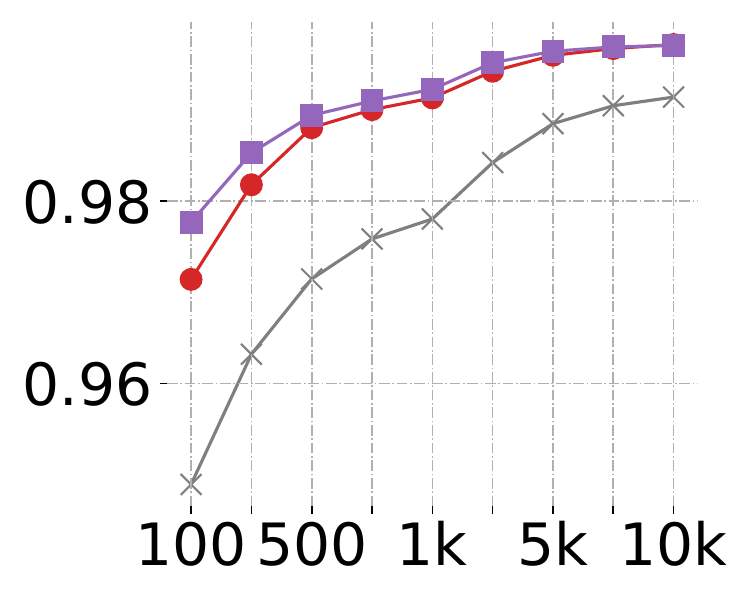}}
  \subfloat[Deep1B\label{fig:re-exp-max-series-d1b}]{
    \includegraphics[width=.1385\textwidth]{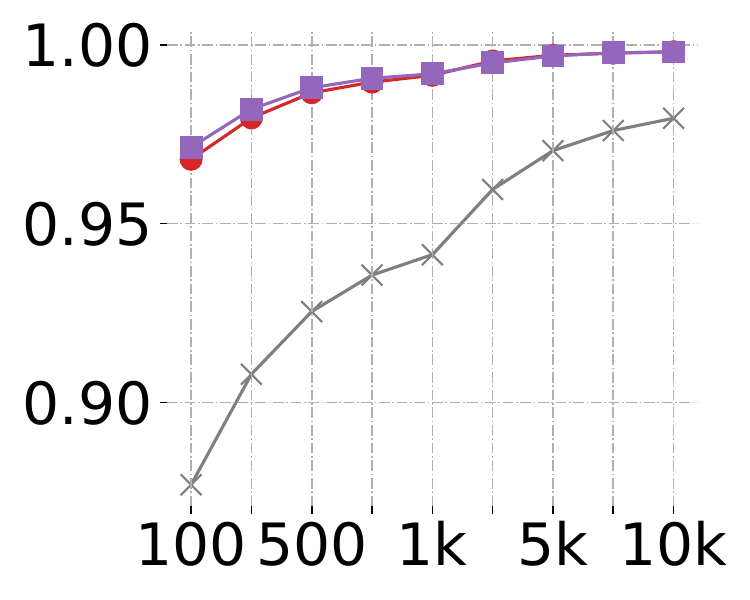}}
  \subfloat[Seismic\label{fig:re-exp-max-series-sm}]{
    \includegraphics[width=.1385\textwidth]{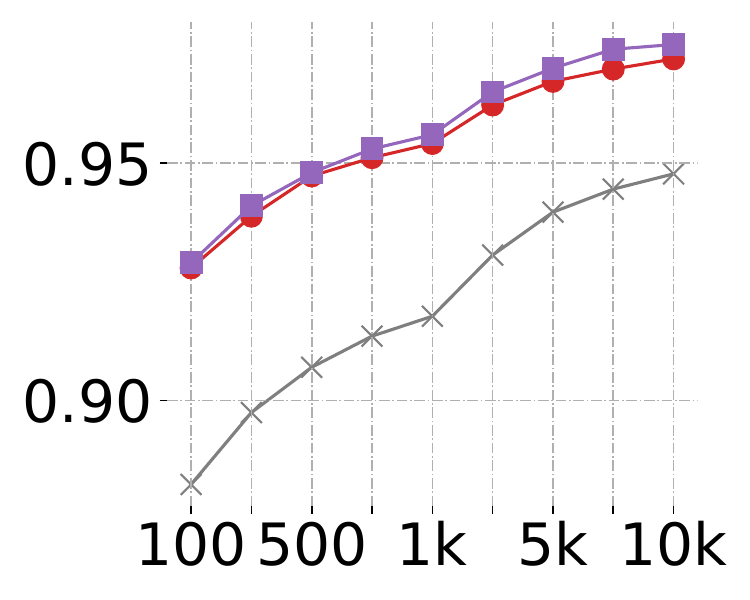}}
  \subfloat[Astro\label{fig:re-exp-max-series-astro}]{
    \includegraphics[width=.1385\textwidth]{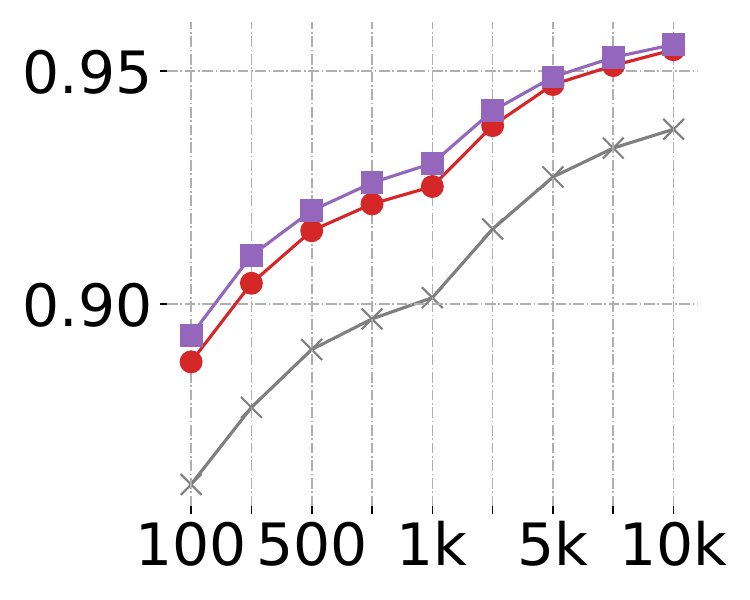}}
  \caption{
    Query answers quality upper-bounds: 1st BSF tightness vs number of series visited; 100M series datasets (higher is better).
    }
  \label{fig:exp-max-series-ideal}
\end{figure*}

\subsection{Time and Convergence}
\label{sec:exp-time}

\noindent{\bf [Training Time]}
We report training times and convergence for all architectures %
in Figure~\ref{fig:exp-convergence}. 
(We omit query answering times because they are proportional to the number of visited series and hence similar among different methods.)
Benefiting from the encoder-only architecture, FDJNet and SEAnet-nD the fastest models to converge in Figure~\ref{fig:exp-sub-time}.
For architectures with both encoder and decoder, SEAnet is 2.9$X$ and 3.35$X$ faster than InceptionTime and TimeNet under similar parameters.
SEAnet converges more steadily to the lowest training loss (Figure~\ref{fig:exp-sub-converge}).

\subsection{SEAsamE}
\label{sec:exp-seasame}

\noindent{\bf [SEAsamE for Similarity Search]}
We evaluate the effectiveness of training SEAnet with SEAsamE for similarity search, by reporting the 1st BSF tightness improvements when compared to SEAsam. 
The 1st BSFs are reported under the constraint that the query answering algorithm examines a maximum of 10,000 series in the index before producing the answer.
The results on 100M datasets, are shown in Figure~\ref{fig:exp-seasame}.
For 16 out of the 31 experiments, %
SEAsamE provided tighter 1st BSFs than SEAsam. %
Its advantage was especially obvious for TimeNet, demonstrating that RNN models can also be promising solutions when combined with specialized sampling and training strategies.

\subsection{SEAtrans}
\label{sec:exp-seatrans}

\noindent{[SEAtrans for Tightness Upper Bounds]}
To further analyze the improvement space of DEA-based data series similarity search, we reported the 1st BSF tightness upper bounds of deploying different summarizations in Figure~\ref{fig:exp-max-series-ideal}.
The 1st BSF tightness \emph{upper bounds} were extracted by examining %
series according to their distances to the query series in the summarization space.
These ideal results reveal the importance of considering the optimization opportunities deriving from both the index structure \emph{and} the summarizations.

  Figure~\ref{fig:exp-max-series-ideal} demonstrates that SEAnet significantly improves the 1st BSF tightness upper bounds achieved by PAA. 
  Moreover, SEAtrans further improves these bounds for four hard real-world datasets, where this is more needed.
  Compared with Figure~\ref{fig:exp-max-series}, these results lead to the following three main observations.
(1) For RandWalk, F5 and F10, all summarizations provided almost perfect upper bounds ($\approx$1).
Hence, the index structure should be better tuned to approach the capabilities of summarizations.
(2) For SALD and Deep1B, DEA learned by SEAtrans provided better upper bounds ($\approx$1) than PAA ($<$1).
In such cases, tuning the index structure is more promising for improving the 1st BSF tightness than tuning the DEA qualities.
(3) For Seismic and Astro, DEA learned by SEAtrans provided better upper bounds ($\approx$96-98\%) than the competitors, but with room for improvement.
This implies that improving the DEA quality on these datasets is crucial in order to further improve similarity search.
These observations are important guidelines for future studies in this direction.

\section{Discussion and Conclusions}\label{sec:conclusions}

In this paper, we introduce the use of deep learning embeddings, DEA, for data series similarity search.
We propose %
novel autoencoders, SEAnet and SEAtrans, designed under the %
SoS preservation principle, for effectively learning DEA.
New sampling strategies, SEAsam and SEAsamE, are introduced in order to facilitate %
training on massive collections.
We demonstrate that the DEA learned by SEAnet and SEAtrans more closely approximates the original data series distances, better preserves the true nearest neighbors in the summarized space, better reconstructs the original series, and leads to better similarity search results than the SOTA PAA-based iSAX. %
These preliminary results are very promising, they set the ground for further advancements in this area, and have the potential to also improve the performance of kNN classification, anomaly detection and other similarity search based applications.

In our future work, we will study the development of lower bounding properties for DEA that will enable exact similarity search, the adaptation of transfer learning or incremental learning techniques for quickly fitting new or dynamic datasets, the development of more powerful sampling strategies, and the careful study of query answering strategies on top of DEA, including \emph{product quantization}~\cite{DBLP:journals/pami/JegouDS11}, \emph{locality sensitive hashing}~\cite{DBLP:conf/icde/LiZSWT020}, and modern %
series indexes~\cite{j20-vldbj-Linardi-ulisse}.

\ifCLASSOPTIONcompsoc
  \section*{Acknowledgments}
\else
  \section*{Acknowledgment}
\fi

Work supported by 
$Y \Pi AI \Theta A$ \& NextGenerationEU project HARSH ($Y\Pi 3TA-0560901$) that is carried out within the framework of the National Recovery and Resilience Plan “Greece 2.0” with funding from the European Union – NextGenerationEU,
IdEx Université Paris Cité %
ANR-18-IDEX-0001, China Scholarship Council, FMJH Program PGMO, %
HIPEAC 4, GENCI–IDRIS (Grants 2020-101471, 2021-101925, 2022-AD011012641R), and NVIDIA Corporation for the Titan Xp GPU donation used in this research.

\ifCLASSOPTIONcaptionsoff
  \newpage
\fi

\bibliographystyle{IEEEtran}
\bibliography{seanet}

\begin{IEEEbiography}[{\includegraphics[width=1in,height=1.25in,clip,keepaspectratio]{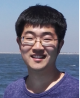}}]{Qitong Wang} 
  is a PhD student in computer science at Université Paris Cité, France, under the supervision of Prof. Themis Palpanas.
  He got his bachelor's and master's degrees in computer science from Fudan University, China.
  His research interests lie in the intersections between massive data series analysis systems and modern machine learning tools.
  He has published several papers in international venues including KDD, VLDB, SIGMOD, DASFAA, etc.
\end{IEEEbiography}

\begin{IEEEbiography}[{\includegraphics[width=1in,height=1.25in,clip,keepaspectratio]{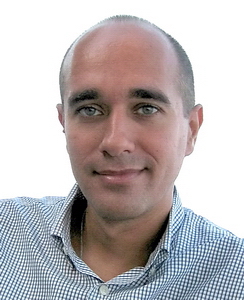}}]{Themis Palpanas}
is Director of the Data Intelligence Institute of Paris (diiP), Senior Member of the French University Institute (IUF), and Professor of computer science at Université Paris Cité (France). %
He is the author of 9 US patents, %
the recipient of 3 Best Paper awards, and the IBM SUR Award. 
He has served as Editor in Chief for BDR Journal, General Chair for VLDB 2013, and Associate Editor for TKDE and DSE journals, as well as PVLDB 2022, 2019 and 2017. %
\end{IEEEbiography}

\end{document}